\documentclass{article}

\usepackage[english]{babel}

\usepackage[letterpaper,top=2cm,bottom=2cm,left=3cm,right=3cm,marginparwidth=1.75cm]{geometry}
\usepackage{appendix}
\usepackage{amsmath,amssymb,amsfonts}
\usepackage{algorithmic}
\usepackage{graphicx}
\usepackage{textcomp}
\usepackage{enumerate}
\usepackage{xcolor}
\usepackage[linktocpage]{hyperref}
\usepackage{tabularx,booktabs}
\usepackage{amsthm}
\usepackage{bbm}
\usepackage{csquotes}
\usepackage{chngcntr}
\usepackage{apptools}
\usepackage{mathtools}
\usepackage{subcaption}
\usepackage{tikz}
\usepackage{caption} 
\usepackage{cleveref}
\usepackage{upgreek}
\usetikzlibrary{mindmap,backgrounds,arrows.meta}
\newcommand{\E}{\mathbb{E}}

\newcommand{\Pm}{\mathcal{P}}

\newcommand{\Q}{\mathcal{Q}}
\newcommand{\X}{\mathcal{X}}

\newcommand{\Y}{\mathcal{Y}}
\newcommand{\F}{\mathcal{F}}
\newcommand{\M}{\mathcal{M}}
\newcommand{\W}{\mathcal{W}}
\newcommand{\Z}{\mathcal{Z}}

\DeclareMathOperator*{\esssup}{ess\,sup}
\DeclareMathOperator*{\essinf}{ess\,inf}

\DeclareMathOperator{\sign}{sign}
\newcommand{\ml}[2]{\mathcal{L}\left(#1  \!\!  \to  \!\!   #2\right)} % maximal leakage
\theoremstyle{plain}
\newtheorem{theorem}{Theorem}
\newtheorem{lemma}{Lemma}
\newtheorem{corollary}{Corollary}
\newtheorem{proposition}{Proposition}
\newtheorem*{lemma*}{Lemma}

\theoremstyle{definition}
\newtheorem{definition}{Definition}

\theoremstyle{remark}

\newtheorem{remark}{Remark}
\newtheorem{example}{Example}
\makeatletter
\newcounter{labelcnt}
\renewcommand{\thelabelcnt}{(\alph{labelcnt})}
\newcommand{\setlabel}[1]{%
	\refstepcounter{labelcnt}\ltx@label{lbl:#1}%
	{\text{\upshape\thelabelcnt}}%
}

\makeatother
\usepackage[sort&compress,square,numbers]{natbib}

\title{Lower Bounds on the Bayesian Risk \\via Information Measures}
\author{Amedeo Roberto Esposito\thanks{Institute of Science and Technology Austria, amedeoroberto.esposito@ist.ac.at}, Adrien Vandenbroucque\thanks{ Ecole Polytechnique F\'ed\'erale de Lausanne, adrien.vandenbroucque@alumni.epfl.ch}, Michael Gastpar\thanks{ Ecole Polytechnique F\'ed\'erale de Lausanne, michael.gastpar@epfl.ch }}

\begin{document}
\maketitle

\begin{abstract}
This paper focuses on  parameter estimation and introduces a new method for lower bounding the Bayesian risk. The method allows for the use of virtually \emph{any} information measure, including R\'enyi's $\alpha$, $\varphi$-Divergences, and Sibson's $\alpha$-Mutual Information. The approach considers divergences as functionals of measures and exploits the duality between spaces of measures and spaces of functions. In particular, we show that one can lower bound the risk with any information measure by upper bounding its dual via Markov's inequality.
        We are thus able to provide estimator-independent impossibility results thanks to the Data-Processing Inequalities that divergences satisfy.
        The results are then applied to settings of interest involving both discrete and continuous parameters, including the ``Hide-and-Seek'' problem, and compared to the state-of-the-art techniques. An important observation is that the behaviour of the lower bound in the number of samples is influenced by the choice of the information measure. We leverage this by introducing a new divergence inspired by the ``Hockey-Stick'' Divergence, which is demonstrated empirically to provide the largest lower-bound across all considered settings. If the observations are subject to privatisation, stronger impossibility results can be obtained via Strong Data-Processing Inequalities. The paper also discusses some generalisations and alternative directions.
\end{abstract}

\tableofcontents

	\section{Introduction}
	In this work,\footnote{. This article was presented in part
at the 2021 and 2022 IEEE International Symposia on Information Theory} we consider the problem of parameter estimation in a Bayesian setting. More precisely, we propose an approach to \emph{lower-bounding} the Bayesian risk leveraging most information measures present in the literature. We look at the problem through an information-theoretic lens, similarly to~\citep{bayesRiskRaginsky}. We thus treat the parameter to be estimated as a message sent through a channel. This allows us to include frameworks where, in a distributed fashion, $m$ processors observe noisy samples of this parameter. The processors will then send a version of their observations to a central node. The central node will then proceed to estimate the parameter. We thus shift the focus from the estimation problem to the computation of two main quantities (which we render as independent of the estimator as possible): 
	\begin{enumerate}
		\item an information measure (\textit{e.g.}, Sibsons's $\alpha$-Mutual Information, $\varphi$-Mutual Information, etc.);
		\item a functional of the probability of some event under independence (\textit{e.g.}, a small-ball probability~\citep{bayesRiskRaginsky});\label{it:smallBall}
	\end{enumerate}
    The main tools utilised rely on Legendre-Fenchel duality and they allow us to introduce bounds involving R\'enyi's, $\varphi$-Divergences and Sibson's $\alpha$-Mutual Information. 
	An advantage of using this type of bounds is that one can render the functional in~\Cref{it:smallBall} (\textit{e.g.}, the small-ball probability) independent of the specific estimator. Similarly, the information measure can also be rendered independent of the estimator via Data-Processing Inequalities. Therefore, these lower bounds can be applied to any standard estimation framework regardless of the specific choice of the estimator. More details on the formal framework that we adopted can be found in~\Cref{sec:bayesianFramework}.
	
	It is important to notice that, although the problem can be interpreted as a transmission problem, a fundamental difference is that the size of the quantised messages may not grow with the number of samples. This might render the reconstruction of the samples impossible but the estimation of the parameter may remain feasible~\citep{bayesRiskRaginsky}. Our main focus will not be on asymptotic results but rather on finite number of samples lower bounds. 
	\subsection{Overview of the document}
    Following the Introduction, the paper will be broken in four main sections:
    \begin{itemize}
        \item  \Cref{sec:preliminaries}:  Preliminaries, in which we will define the information measures of interest as well as describe the theoretical framework leveraged to provide the bounds;
        \item \Cref{sec:riskLBThroughProbs}: Main Bounds, in which (making use of the famework described in~\Cref{sec:preliminaries}) we propose a variety of lower bounds on the Bayesian Risk involving a variety of information measures, in particular:
        \begin{itemize}
            \item Sibson's $\alpha$-Mutual Information and Maximal Leakage (\Cref{thm:sibsMIResultBayesRisk});
            \item $\varphi$-Mutual Information (\Cref{thm:fDivBoundBayesRisk}). In particular:
            \begin{itemize}
                \item Hellinger $p$-Divergence (\Cref{thm:lowerBoundHellinger});
                \item R\'enyi's $\alpha$-Divergence (\Cref{rmk:lowerBoundRenyiDivergence});
                \item a generalisation of the ``Hockey-Stick'' Divergence,  $E_\gamma$ (\Cref{thm:lowerBoundGeneralHockeyStick}); 
            \end{itemize}
        \end{itemize}
        \item \Cref{sec:bayesRiskExample}: Examples, in which we apply the bounds proposed in~\Cref{sec:riskLBThroughProbs} to a variety of classical and less classical settings:
        \begin{itemize}
            \item estimation of the bias of a Bernoulli random variable (see~\Cref{sec:ex1});
            \item estimation of the bias of a Bernoulli random variable after injection of additional noise (\textit{e.g.}, observing privatised samples, see~\Cref{sec:noisy_bernoulli_bias});
            \item estimation of the mean of a Gaussian random variable (with Gaussian prior, see~\Cref{sec:gaussianExample});
            \item lower-bound on the minimax risk for the ``Hide-and-Seek'' problem~\cite{nipsHideAndSeek} (see~\Cref{sec:hideAndSeek}).
        \end{itemize}
        For each of the problems we derive bounds involving a variety of information measures and we compare said bounds among themselves and with respect to relevant bounds in the literature as well;
        \item throughout the document we also consider further generalisations, in which we propose a variety of ways of extending/tightening/altering the results we proposed in~\Cref{sec:riskLBThroughProbs}. In particular, one can provide new bounds:
        \begin{itemize}
            \item conditioning on an additional (and cleverly constructed) random variable (see~\Cref{sec:conditioning});
            \item leveraging the asymmetry of some information measures (see~\Cref{sec:invertingBayesRisk});
            \item ignoring the probabilities and lower-bounding the risk directly (see~\Cref{sec:directLowerBoundRisk}).
        \end{itemize}
    \end{itemize}
	\subsection{Related Work}
	The problem of parameter estimation has been extensively studied over the years, with many contributions coming from a variety of fields. Relevant literature, mostly leveraging the Van Trees Inequality (and the quadratic risk) can be found in ~\cite{estimationVanTrees, estimationVanTrees1, estimationVanTrees2,estimationVanTrees3, estimationVanTrees4}.  
	Moreover, a survey of early work in this area (mainly focusing on asymptotic settings) can be found in~\cite{han}. More recent but important advances are instead due to~\cite{duchiEstimationIT, duchiFano2013, nipsHideAndSeek}.
    Closely connected to this work is \cite{bayesRiskRaginsky}. The approach is quite similar, with the main difference that we employ a family of bounds involving a variety of divergences while~\cite{bayesRiskRaginsky} relies solely on Mutual Information and the Kullback-Leibler Divergence. Related is also~\cite{CalmonSDPIHockeyStick}, where the authors use the so-called $E_\gamma$-Divergence to provide a lower-bound on the Bayesian Risk. A similar approach was also undertaken in~\cite{bayesRiskFInformativity}. The authors focused on the notion of $\varphi$-informativity (cf.~\cite{fInformativity}) and leveraged the Data-Processing inequality similarly to~\cite[Theorem 3]{fullVersionGeneralization}. A more thorough discussion of the differences between this work and~\cite{bayesRiskFInformativity} can be found in~\Cref{app:comparisonWithChen}. %In particular, $f$-informativities are more general than the $f$-Mutual Informations considered in this work (cf. Definition~\ref{def:fMI}) and they can potentially lead to tighter results. The technique used to provide lower bounds on the Bayesian risk for general non-negative losses (cf.~\cite[Section 4]{bayesRiskFInformativity}) is, however, different. It is unclear whether the results provided in this work are equivalent (or weaker) with respect to those obtained in~\cite{bayesRiskFInformativity}. 
	\subsection{Problem Setting} \label{sec:bayesianFramework}
	Let $\mathcal{W}$ denote the parameter space and assume that we have access to a prior distribution over $\mathcal{W}$, $\mathcal{P}_W$. Suppose that we observe $W\sim\Pm_W$ through the family of distributions $\mathcal{P}= \{ \mathcal{P}_{X|W=w}: w\in\mathcal{W} \}.$ Given a function $\phi:\mathcal{X}\to\hat{\mathcal{W}}$, one can then estimate $W$ from $X\sim \mathcal{P}_{X|W}$ via $\phi(X)=\hat{W}$. Let us denote with $\ell:\mathcal{W}\times\hat{\mathcal{W}}\to \mathbb{R}^+$ a loss function, the Bayesian risk is defined as:
	\begin{equation}
		R_B= \inf_\phi\,\Pm_{W\hat{W}}(\ell(W,\phi(X))) = \inf_\phi\,\Pm_{W\hat{W}}(\ell(W,\hat{W})).\label{risk}
	\end{equation}
	Our purpose is to lower-bound $R_B$ using a variety of information measures. With this drive and leveraging various tools that stem from Legendre-Fenchel duality, one can connect the expected value of $\ell$ under the joint  $\Pm_{W\hat{W}}$ to 
	\begin{itemize}
		\item the  expected value of the same function under the product of the marginals ($\Pm_W\Pm_{\hat{W}}$) or a ``small-ball probability'' ;
		\item an information measure (quantifying how ``far'' the joint is from the product of the marginals).
	\end{itemize}
	%To this end, among other tools, we will be using a simple Markov's inequality approach: \textit{i.e.}, for every estimator $\phi$ and $\rho \geq 0$, one can do the following
	%\begin{equation}
	%    \Pm_{W\hat{W}}
	%    \left(\ell(W,\hat{W})\right) \geq \rho\left(\Pm_{W\hat{W}}(\ell(W,\hat{W})\geq \rho)\right). \label{markov}
	%\end{equation}
	This will allow us to render the lower-bound \textit{as independent as possible} from the  specific choice of the estimator $\phi$. 
	%For instance, a family of suitable results is the one involving Sibson's $\alpha$-Mutual Information, as it naturally allows us to ``substitute'' $\Pm_{W\hat{W}}(E)$ with a function of $\max_{\hat{w}}\Pm_W(E_{\hat{w}})$ (that will thus depend on $\phi$ only through the support of $\hat{W}$) and the Sibson's Mutual Information $I_\alpha(W,\hat{W})$.
	More precisely, our desideratum will be a lower-bound of the following form:
	\begin{equation}
		R_B\geq \upvarpi\left(\frac{d\Pm_{W\hat{W}}}{d\Pm_W\Pm_{\hat{W}}}\right)\upvartheta(\Pm_W\Pm_{\hat{W}},\ell), \label{eq:desideratumBayesRisk}
	\end{equation}
	with, once again, the purpose of then rendering the right-hand side of~\Cref{eq:desideratumBayesRisk} as independent as possible of the estimator $\phi$. 
	Let us denote with $L_W(\hat{W},\rho)=\Pm_W \Pm_{\hat{W}}(\ell(W,\hat{W})< \rho)$, a functional $\upvartheta$ of particular interest to us is the one that leads to (a function of) the so-called small-ball probability  \begin{equation}L_W(\rho)= \sup_{\hat{w}\in\mathcal{\hat{W}}} L_W(\hat{w},\rho) =  \sup_{\hat{w}\in\mathcal{\hat{W}}}\Pm_W(\ell(W,\hat{w})< \rho).\label{smallBall}\end{equation}
	More generally, the choice of $\upvartheta$ will depend on the choice of $\upvarpi$ and vice versa. In the following sections, we will explore different choices of said functionals that lead to interesting results in the field.
	\section{Preliminaries}\label{sec:preliminaries}
	In this section, we will define the main objects utilised throughout the document and define the relevant notation.
	We will mainly adopt a measure-theoretic framework.
	Given a measurable space $(\X,\mathcal{F})$ and two measures $\mu,\nu$ which render it a measure space, if $\nu$ is absolutely continuous with respect to $\mu$ (denoted with $\nu\ll\mu$) then we will represent with $\frac{d\nu}{d\mu}$ the Radon-Nikodym derivative of $\nu$ with respect to $\mu$. Given a (measurable) function $f:\X \to \mathbb{R}$ and a measure $\mu$ we denote with $\mu(f) = \langle f, \mu \rangle =  \int f d\mu$ \textit{i.e.}, the 
Lebesgue integral of $f$ with respect to the measure $\mu$ represents a bilinear inner-product which will characterise a pairing between a (properly defined) space of functions and a (properly defined) space of measures. Once such a pairing is set, one can then proceed onto defining the Legendre-Fenchel transform connecting functionals acting on measures to functionals acting on functions. More formally, let $C_b(\X)$ denote the space of continuous and bounded functions defined on $\X$ and $\mathcal{M}(\X)$ the set of Radon measures defined on the same space, then one has that $\mathcal{M}(\X)$ and $C_b(\X)$ are in separating duality through the bilinear mapping $\langle \cdot , \cdot \rangle$ defined above (see~\cite{largeDeviationConvexAnalysis}). Thus, given a functional $\psi : C_b(\X)\to\mathbb{R}$ one can define its Legendre-Fenchel dual as follows:
	\begin{equation}
		\psi^\star(\mu) = \sup_{f\in C_b(\X)} \langle f,\mu \rangle - \psi(f). \label{eq:legendreFenchelGeneral}
	\end{equation}
%	The Legendre-Fenchel dual represent  a standard way of connecting  (a functional acting on)  a space with (a functional acting on)  its dual. 
Another connection between the spaces of interest comes from considering a norm on a space and the corresponding dual norm on the dual space, \textit{i.e.,} given a norm acting on $\X$, $\left\lVert  \cdot \right\rVert$ and a pairing between two spaces $(\X,\Y)$, one can construct a norm on $\Y$ as follows:
	\begin{equation}
		\left\lVert h \right\rVert_\star = \sup_{f: \left\lVert  f \right\rVert \leq 1 } |\langle h,f \rangle |. \label{eq:dualNorm} \end{equation}
	For  the purpose of  this paper, we will essentially interpret the expected value $\Pm_{W\hat{W}}(
	\ell)$ as  $\langle \Pm_{W\hat{W}}, \ell\rangle$. Once this simple observation is made, these tools  will allow us to connect functionals of measure (\textit{e.g.}, information measures) to functionals of the loss (\textit{e.g.}, small-ball probabilities).
	\subsection{Information Measures}\label{infoMeas}
	We will now continue introducing the information measures that we need in order to provide the main results.
	\subsubsection{R\'enyi's $\alpha$-Divergence}\label{sec:renyisDiv}
	Introduced by R\'enyi as a generalization of KL-divergence, $\alpha$-divergence has found many applications ranging from hypothesis testing to guessing and several other statistical inference problems~\cite{RenyiKLDiv}. Indeed, it has several useful operational interpretations (e.g., hypothesis testing, and the cut-off rate in block coding \cite{RenyiKLDiv,opMeanRDiv1}).
	It can be defined as follows~\cite{RenyiKLDiv}.
	\begin{definition}
		Let $(\Omega,\F,\Pm),(\Omega,\F,\Q)$ be two probability spaces. Let $\alpha>0$ be a positive real number different from $1$. Consider a measure $\mu$ such that $\Pm\ll\mu$ and $\Q\ll\mu$ (such a measure always exists, e.g. $\mu=(\Pm+\Q)/2$)) and denote with $p,q$ the densities of $\Pm,\Q$ with respect to $\mu$. The $\alpha$-Divergence of $\Pm$ from $\Q$ is defined as follows:
		\begin{align}
			D_\alpha(\Pm\|\Q)=\frac{1}{\alpha-1} \log \int p^\alpha q^{1-\alpha} d\mu.
		\end{align}
	\end{definition}
	\begin{remark}
		The definition is independent of the chosen measure $\mu$. %whenever $\infty>\alpha>0$ and $\alpha\neq 1$. 
		It is indeed possible to show that
		$\int p^{\alpha}q^{1-\alpha} d\mu = \int \left(\frac{q}{p}\right)^{1-\alpha}d\Pm $, and that whenever $\Pm\ll\Q$ or $0<\alpha<1,$ we have $\int p^{\alpha}q^{1-\alpha} d\mu= \int \left(\frac{p}{q}\right)^{\alpha}d\Q$, see \cite{RenyiKLDiv}.
	\end{remark}
	
	It can be shown that if $\alpha>1$ and $\Pm\not\ll\Q$ then $D_\alpha(\Pm\|\Q)=\infty$. The behavior of the measure for $\alpha\in\{0,1,\infty\}$ can be defined by continuity. In general, one has that $D_1(\Pm\|\Q) = D(\Pm\|\Q)$, where $D(\Pm\|\Q)$ denotes the Kullback-Leibler divergence between $\Pm$ and $\Q$. However, if $D(\Pm\|\Q)=\infty$ or there exists $\beta$ such that $D_\beta(\Pm\|\Q)<\infty$ then $\lim_{\alpha\downarrow1}D_\alpha(\Pm\|Q)=D(\Pm\|\Q)$~\cite[Theorem 5]{RenyiKLDiv}. 
	For an extensive treatment of $\alpha$-divergences and their properties we refer the reader to~\cite{RenyiKLDiv}. 
	Starting from R\'enyi's Divergence and the geometric averaging that it involves, Sibson built the notion of Information Radius \cite{infoRadius}. A deconstructed and generalised version of the Information Radius leads us to the following definition of a generalisation of Shannon's Mutual Information~\cite{verduAlpha}:
	\begin{definition}\label{SibsonsAlpha}
		Let $X$ and $Y$ be two random variables jointly distributed according to $\Pm_{XY}$, and with marginal distributions $\Pm_X$ and $\Pm_Y$, respectively. For $\alpha>0$, the Sibson's mutual information of order $\alpha$ between $X$ and $Y$ is defined as:
		\begin{align}
			I_\alpha(X,Y) = \min_{\Q_Y} D_\alpha(\Pm_{XY}\|\Pm_X \Q_Y) \label{iAlphaDef}.
		\end{align}
	\end{definition}
	The following, alternative formulation is also useful \cite{verduAlpha}:
	\begin{align}
		I_\alpha(X,Y)&= \frac{\alpha}{\alpha-1}\log\mathbb{E}\left[\mathbb{E}^{\frac{1}{\alpha}}\left[\left(\frac{\Pm_{Y|X}}{\Pm_Y}\right)^\alpha \bigg|Y \right]\right] \label{altFormulSibs} \\
		&= D_\alpha(\Pm_{XY}\|\Pm_X \Pm_Y) - D_\alpha(\Pm_{Y_\alpha}\|\Pm_Y),
	\end{align}
	where $\Pm_{Y_\alpha}$ is the measure minimizing \eqref{iAlphaDef}. In analogy with the limiting behavior of $\alpha$-Divergence we have that $\lim_{\alpha\to 1}I_\alpha(X,Y)=I(X;Y)$, where $I(X;Y)$ represents the Mutual Information between $X$ and $Y$. When $\alpha\to\infty$ we retrieve the following object: $$I_\infty(X,Y)=\log\mathbb{E}_{\Pm_Y}\left[\sup_{x:\Pm_X(x)>0} \frac{\Pm_{XY}(x,Y)}{\Pm_X(x)\Pm_Y(Y)}\right].$$
	To conclude, let us list some of the properties of the measure:
	\begin{proposition}[\cite{verduAlpha}]
		\label{sibsProperties}
		\noindent
		\begin{enumerate}
			\item \textbf{Data-Processing Inequality}:
			given $\alpha>0$, 
			$I_{\alpha}(X,Z) \leq \min\{I_{\alpha}(X,Y),I_{\alpha}(Y,Z)\}$ if the Markov Chain $X-Y-Z$ holds;
			\item $I_\alpha(X,Y)\geq 0$ with equality iff $X$ and $Y$ are independent;
			\item Let $\alpha_1\leq \alpha_2$ then $I_{\alpha_1}(X,Y)\leq I_{\alpha_2}(X,Y)$;
			\item Let $\alpha\in (0,1)\cup(1,\infty)$, for a given $\Pm_X$, $\frac{1}{\alpha-1}\exp\left(\frac{\alpha-1}{\alpha} I_\alpha(X,Y)\right)$ is convex in $\Pm_{Y|X}$;
			\item $I_\alpha(X,Y) \leq \min\{\log |X |, \log |Y|\}$;
		\end{enumerate}
	\end{proposition}
	For an extensive treatment of Sibson's $\alpha$-MI we refer the reader to \cite{verduAlpha}.
	\subsubsection{Maximal Leakage}\label{maximalLeakage}
	A particularly relevant dependence measure, strongly connected to Sibson's Mutual Information is the maximal leakage, denoted by $\ml{X}{Y}.$ It was introduced as a  way of measuring the leakage of information from $X$ to $Y$, hence the following definition:
	\begin{definition}[Def. 1 of \cite{leakageLong}]\label{leakage}
		Given a joint distribution $\Pm_{XY}$ on finite alphabets $\mathcal{X}$ and $\mathcal{Y}$, the maximal leakage from $X$ to $Y$ is defined as:
		\begin{equation}  \ml{X}{Y}= \sup_{\substack{U-X-Y-\hat{U}}} \log \frac{\mathbb{P}(\{U=\hat{U}\})}{\max_{u\in\mathcal{U}} \mathbb{P}_U(\{u\})},\end{equation}
		where $U$ and $\hat{U}$ take values in the same finite, but arbitrary, alphabet. 
	\end{definition} \noindent
	It is shown in~\cite[Theorem 1]{leakageLong} that, for finite alphabets: 
	\begin{align} \label{leakageFormula}
		\ml{X}{Y} = \log \sum_{y\in \mathcal{Y}} \max_{\substack{x\in\mathcal{X}}: \Pm_{X}(x)>0} \Pm_{Y|X}(y|x). \end{align}
	If $X$ and $Y$ have a jointly continuous pdf $f_{XY}(x,y)$, we get~\cite[Corollary 4]{leakageLong}:
	\begin{align}\ml{X}{Y} = \log \int_{\mathbb{R}} \sup_{x: f_X(x)>0} f_{Y|X}(y|x) dy. \label{leakageContinous} \end{align}
	One can show that $\ml{X}{Y}= I_\infty(X;Y)$ i.e., Maximal Leakage corresponds to the Sibson's Mutual Information of order infinity.  This allows the measure to retain the properties listed in Proposition \ref{sibsProperties}, furthermore: 
	\begin{lemma}[\cite{leakageLong}]
		\label{leakageProps}
		For any joint distribution $\Pm_{XY}$ on finite alphabets $\X$ and $\Y$, $\ml{X}{Y} \geq I(X; Y )$.
	\end{lemma}
	Another relevant notion is Conditional Maximal Leakage: 
	\begin{definition}[Conditional Maximal Leakage \cite{leakageLong}]
		Given a joint distribution $\mathcal{P}_{XYZ}$ on alphabets $\mathcal{X}, \mathcal{Y}, \text{ and } \mathcal{Z}$, define:
		\begin{equation} \ml{X}{Y|Z} =\sup_{\substack{U:U-X-Y|Z}} \log\frac{\mathbb{P}(\{U = \hat{U}(Y, Z)\})}{\mathbb{P}(\{U=\tilde{U}(Z)\})},\end{equation}
		where $U$ takes value in an arbitrary finite alphabet and we consider $\hat{U}, \tilde{U}$ to be the optimal estimators of $U$ given $(Y,Z)$ and $Z$, respectively.
	\end{definition}
	\noindent Again, it is shown in~\cite{leakageLong} that for discrete random variables $X,Y,Z$:
	\begin{equation} \ml{X}{Y|Z} = \log  \max_{\substack{z:\mathcal{P}_Z(z)>0}} \sum_{y} \max_{\substack{x: \mathcal{P}_{X|Z}(x|z)}>0} \mathcal{P}_{Y|XZ}(y|xz),\notag
	\end{equation}
	and \begin{equation}\ml{X}{(Y,Z)} \leq \ml{X}{Y} + \ml{X}{Z|Y}. \label{ineqLeak}
	\end{equation}
	
	\subsubsection{$\varphi$-Mutual Information}\label{fDivergence}
	Another generalization of the KL-Divergence can be obtained by considering a generic convex function $\varphi:\mathbb{R}^+\to \mathbb{R}$, usually with the simple constraint that $\varphi(1)=0$. The constraint can be ignored as long as $\varphi(1)<+\infty$ by simply considering a new mapping $\tilde{\varphi}(x) = \varphi(x) - \varphi(1)$. 
	\begin{definition}
		Let $(\Omega,\F,\Pm),(\Omega,\F,\Q)$ be two probability spaces. Let $\varphi:\mathbb{R}^+\to \mathbb{R}$ be a convex function. Consider a measure $\mu$ such that $\Pm\ll\mu$ and $\Q\ll\mu$. Denoting with $p,q$ the densities of the measures with respect to $\mu$, the $\varphi$-Divergence of $\Pm$ from $\Q$ is defined as follows:
		\begin{align}
			D_\varphi(\Pm\|\Q)=\int q \varphi\left(\frac{p}{q}\right) d\mu.
		\end{align}
	\end{definition}
	Despite the fact that the definition uses $\mu$ and the densities with respect to this measure, it is possible to show that $\varphi$-Divergences are actually independent from the dominating measure \cite{fDiv1}. Indeed, when absolute continuity between $\Pm,\Q$ holds, i.e. $\Pm\ll\Q$, an assumption we will often use, we retrieve the following \cite{fDiv1}:
	\begin{equation}
		D_\varphi(\Pm\|\Q)= \int \varphi\left(\frac{d\Pm}{d\Q}\right)d\Q.
	\end{equation}Denoting with $\F_X$ the $\sigma$-field generated from the random variable $X$, (i.e., $\sigma(X)$), $\varphi$-Mutual Information is defined as follows:
	\begin{definition}\label{def:fMI}
		Let $X$ and $Y$ be two random variables jointly distributed according to $\Pm_{XY}$ over the measurable space $(\X\times\Y, \F_{XY})$. 
		Let  $(\X,\F_{X},\Pm_{X}),(\Y,\F_{Y},\Pm_Y)$ be the corresponding probability spaces induced by the marginals.  Let $\varphi:\mathbb{R}^+\to \mathbb{R}$ be a convex function such that $\varphi(1)=0$. The $\varphi$-Mutual Information between $X$ and $Y$ is defined as:
		\begin{equation}
			I_\varphi(X,Y)=D_\varphi(\Pm_{XY}\|\Pm_X\Pm_Y).
		\end{equation}
		If $\Pm_{XY}\ll\Pm_X\Pm_Y$ we have that:
		\begin{equation}
			I_\varphi(X,Y) = \int \varphi\left(\frac{d\Pm_{XY}}{d\Pm_X\Pm_Y}\right)d\Pm_X\Pm_Y.
		\end{equation}
	\end{definition}
	It is possible to see that, if $\varphi$ satisfies $\varphi(1)=0$ and it is strictly convex at $1$, then  $I_\varphi(X,Y)=0$ if and only if $X$ and $Y$ are independent \cite{fDiv1}.
	This generalisation includes the Kullback-Leibler Divergence (by simply setting $\varphi(t)=t\log(t)$) and allows to retrieve $\alpha$-Divergences through a one-to-one mapping. Other meaningful examples are: the Total Variation distance ($\varphi(t)=\frac12|t-1|$), the Hellinger distance ($\varphi(t)=(\sqrt{t}-1)^2$), Pearson's $\chi^2$-divergence ($\varphi(t)=(t-1)^2$), etc.
	Exploiting a bound involving $I_\varphi(X,Y)$ for a broad enough set of functions $\varphi$ allows to differently measure the dependence between $X$ and $Y$.  This allows us to provide bounds that are tailored to the specific problem at hand and, as we will see, to provide bounds that are tighter with respect to the ones leveraging the well-known Mutual Information.
    \subsection{(Strong) Data-Processing Inequalities}\label{sec:SDPI}
     An important property that $\varphi$-Divergences share is the Data-Processing Inequality (DPI) \textit{i.e.}, given two measures $\mu,\nu$ and a Markov Kernel $K$, one has that for every convex $\varphi$
    \begin{equation}
        D_\varphi(\nu K\| \mu K) \leq D_\varphi(\nu\|\mu).\label{eq:DPI}
    \end{equation}
    This property holds as well for R\'enyi's $\alpha$-Divergences, despite them not being a $\varphi$-Divergence~\cite[Theorem 9]{RenyiKLDiv}.
    DPIs represent a powerful and widely used tool to derive bounds in the field, so much so that a large body of work has focused on tightening said inequalities. In particular, in many settings of interest, given a reference measure $\mu$ one can show that $D_\varphi(\nu K\| \mu K)$ is strictly smaller than $D_\varphi(\nu\|\mu)$  unless $\nu=\mu$ and the characterisation of the ratio between these two quantities for Markov Kernels of interest  gave rise to the concept and study of ``strong Data-Processing Inequalities''. Formally, 
\begin{definition}[{\cite[Definition 3.1]{sdpiRaginsky}}] Given a probability measure $\mu$, a Markov Kernel $K$ and  a convex function $\varphi,$ we say that $K$ satisfies a $\varphi$-type strong data-processing inequality (SDPI) at $\mu$ with constant $c\in[0,1)$ if, for all $\nu\ll\mu$ one has that
    \begin{equation}
    D_\varphi(\nu K\|\mu K) \leq c\cdot D_\varphi(\nu\|\mu).
    \end{equation}
    In order to characterise the tightest such constant $c$, let us define the following objects:
\begin{align*}
\eta_\varphi(\mu,K) &= \sup_{\nu\neq \mu} \frac{D_\varphi(\nu K\|\mu K) }{D_\varphi(\nu\|\mu)} \\
\eta_\varphi(K) &= \sup_{\mu} \eta_\varphi(\mu,K).
\end{align*}
\end{definition}
Said quantities are generally hard to compute for a given Markov Kernel $K$ and functional $\varphi$, however a variety of bounds is present in the literature (see~\cite{sdpiRaginsky}).
In particular, given any convex $\varphi$ it is possible to show the following result:
	\begin{lemma}[{\cite[Theorem 4.1]{cohenDobrushin}}]
		Let $K:\mathcal{F}\times\Omega\to[0,1]$ represent a Markov Kernel, and let $\varphi$ be a convex functional such that $\varphi(1)=0$, one has that
		\begin{equation}
			\eta_\varphi(K) \leq \vartheta(K),\label{eq:SDPIDobrushin}
		\end{equation}
		where $\vartheta(K)$ represents the Dobrushin contraction coefficient of $K$ \textit{i.e.},
		\begin{equation}
			\vartheta(K) = \max_{x,\hat{x}\in\Omega} TV( K(\cdot|x),K(\cdot|\hat{x})).
		\end{equation}
		
	\end{lemma}
 Despite the generality of the result one can show that~\Cref{eq:SDPIDobrushin} does not hold true for R\'enyi's Divergences:
 \begin{example}
     Let $\mu=(1/2,1/2)$ and $K=\text{BSC}(\lambda)$ with $\lambda<\frac12$. Then $\eta_{TV}(K)=\vartheta(K)=(1-2\lambda)$. 
     Consider now $D_\alpha(K(\cdot|0)\|\mu)=D_\alpha(\delta_0 K \|\mu K) = \frac{1}{\alpha-1} \log ( 2^{1-\alpha}( \lambda^\alpha +(1-\lambda)^\alpha))$. Moreover, $D_\alpha(\delta_0\| \mu) = -\log(2)$. If $\lambda = 0.2$ and $\alpha=6$ one has that
     \begin{equation}
         \eta_{D_\alpha}(K) > \frac{D_\alpha(\delta_0 K \|\mu K)}{D_\alpha(\delta_0\| \mu)} = 0.6138 > \vartheta(K) = 0.6,
     \end{equation}
     and the gap increases with $\alpha.$
 \end{example}
 Moreover, a universal lower-bound is known as well, in case of three times differentiable functions $\varphi$ with $\varphi''(1)>0$~\cite[Theorem 3.3]{sdpiRaginsky}:
 \begin{equation}
     \eta_\varphi(K) \geq \eta_{\chi^2}(K) = \sup_{\mu} \left(\sup_{f,g} \mu K(f\cdot g)\right)^2=\sup_{\mu} S^2(\mu,K), \label{eq:maxCorrLowerBound}
 \end{equation}
 where the term on the right-hand side of~\Cref{eq:maxCorrLowerBound} is also known in the literature as ``Maximal Correlation''. For functions $\varphi$ which are operator convex, the inequality in~\Cref{eq:maxCorrLowerBound} is actually an equality. Examples of operator convex functions are: $x\log(x)$, $x^p$ with $1\leq p\leq 2$, and so on. This is particularly advantageous because it allows us to leverage tensorisation properties of SDPI that are known to hold for the Kullback-Leibler Divergences even for other divergences (\textit{e.g.}, Hellinger $p$ with $p\leq 2$). More precisely, one can say the following: if $\varphi$ is operator convex and the channel considered is an $n$-fold tensor-product $K^{\otimes n}$ then the following holds true~\cite[Corollary 3.1]{sdpiRaginsky},\cite[Equation (62)]{SDPIfDiv}:
 \begin{equation}
     \eta_\varphi(K^{\otimes n}) \leq 1-(1-\eta_\varphi(K))^n. \label{eq:tensorisationOperatorConvex}
 \end{equation}
 
 \subsection{Variational Repesentations and Functional Inequalities}\label{sec:varReprDiv}
	A re-interpretation of the comments stated in~\Cref{sec:preliminaries} and tailored to divergences leads us to the main technical tools that will be used through the document: ``variational representations'' and functional inequalities. The main starting point will be looking at divergences as functionals acting on the first measure \textit{i.e.},  $D_\varphi(\cdot\|\mu)=\psi_\mu(\cdot)$. Once this is established, most variational representations are instances of Legendre-Fenchel duality as stated in~\Cref{eq:legendreFenchelGeneral}.
	The most well-known is certainly the Donsker-Varadhan representation of the Kullback-Leibler Divergence, that states the following~\cite{largeDeviationVaradhan}:
	\begin{equation}
		D(\nu\|\mu) = \sup_{f\in B(\X)} \langle \nu, f \rangle - \log\left(\mu(\exp(f))\right), \label{eq:donskerVaradhan}
	\end{equation}
	where $B(\X)$ denotes the space of bounded and measurable real-valued functions defined on $\X$.~\Cref{eq:donskerVaradhan} characterises the Kullback-Leibler divergence as the Legendre-Fenchel dual of the functional $\log(\mu(\exp(f)))=\vartheta_\mu(f)$ \textit{i.e.}, $D(\cdot\|\mu)=\vartheta_\mu^\star(\cdot).$
	Similar variational representation can be found for large families of Divergences, like R\'enyi's Divergences~\cite{anantharamRenyiDivVarRepr,RenyiDivVarReprNeuralNetwork} and $\varphi$-Divergences~\cite{minimizationMeasures}. We will now lay the ground-work in order to state the variational representation for $\varphi$-Divergences as it represents a meaningful tool for the scope of this work.
	In particular, let $F(\X)$ be an arbitrary family of real-valued functions defined on $\X$ and denote with $\mathcal{M}_1(\X)$ the space of probability measures over $\X$. Denote with $\langle F(\X)\cup B(\X) \rangle$ the linear span of $F(\X)\cup B(\X)$. Moreover, consider the following sets:
	$$ \mathcal{M}_1^F(\X)= \left\{\nu \in \mathcal{M}_1(\X) : \int |f|d\nu <\infty \text{ for } f \in F(\X) \right\},$$
	and
	$$ \mathcal{M}^F(\X)= \left\{\nu \in \mathcal{M}(\X): \int |f|d|\nu|\text{\footnotemark}<\infty \text{ for } f \in F(\X) \right\}.$$
	\footnotetext{$|\nu|$ denotes the total variation of the finite signed measure $\nu$. }
	If $F(\X)=B(\X)$ then $\mathcal{M}_1^F(\X)= \M_1(\X)$ and $\M^F(\X)=\M(\X)$.
	Denote with $\tau_F$ the weakest topology on $\M^F(\X)$ such that all mappings $\nu\to\nu(f)$ are continuous when $f\in \langle F(\X)\cup B(\X) \rangle$ and with $\tau_M$ the weakest topology on $\langle F(\X)\cup B(X) \rangle$ such that all mappings $f \to \nu(f)$ are continuous when $\nu\in\M^F(\X)$. 
	One can then show the following result
	\begin{proposition}[{\cite[Proposition 2.1]{minimizationMeasures}}]
		The space $\M^F(\X)$ equipped with the $\tau_F$-topology and the space $\langle F(\X)\cup B(\X) \rangle$ equipped with the $\tau_M$ are locally convex topological vector spaces and are the topological dual of each other.
	\end{proposition}
	$D_\varphi(\cdot\|\mu) = \psi_\mu(\cdot)$ is thus a convex and lower semi-continuous mapping with respect to $\tau_F$~\cite[Proposition 2.2]{minimizationMeasures} and it is possible to characterise its variational representation, bridging us between the two spaces $\M^F(\X)$ and $\langle F(\X)\cup B(\X) \rangle$. For the additional technical condition required on $\varphi$ (\textit{i.e}, guaranteeing the uniqueness of the dual optimal solution the reader is referred to~\citep{minimizationMeasures}).
\begin{theorem}[{{\citep[Proposition 4.2 and Theorem 4.3]{minimizationMeasures}}}]\label{thm:varReprfDiv}
	Let $\varphi$ be a strictly convex functional and let $\mu \in \mathcal{M}(\X)$. One has that for every $\nu\in\M^F(\X)$:
	\begin{equation}
		D_\varphi(\nu\|\mu) = \sup_{f\in\langle F(\X)\cup B(\X) \rangle} \nu(f) - \mu(\varphi^\star(f)),\label{eq:varReprfDiv}
	\end{equation}
	where $\varphi^\star$ denotes the Legendre-Fenchel dual of $\varphi$.
	Moreover, one has that for a given $f\in \langle F(\X)\cup B(\X) \rangle$:
	\begin{equation}
		\mu(\varphi^\star(f)) = \sup_{\nu\in\M^F(\X)} \nu(f)-D_\varphi(\nu\|\mu). \label{eq:varReprDualFDiv}
	\end{equation}
\end{theorem}
\begin{remark}[Expected values, divergences and duality]\label{rmk:onTheDual}
Through~\Cref{eq:varReprfDiv}, given a measure $\mu$, one can connect the expected value of any function $f\in \langle F(\X)\cup B(\X) \rangle$ under any measure $\nu\ll\mu$ (\textit{i.e.}, $\nu(f)$) to the divergence $D_\varphi(\nu\|\mu)$. The behavior of the third actor in~\Cref{eq:varReprfDiv}, the dual of $D_\varphi(\cdot\|\mu)$, is crucial in order to obtain bounds. For instance, when $f$ is the indicator function of an event, one can explicitly compute the dual (and then retrieve a family of Fano-like inequalities involving arbitrary divergences, see~\cite{fullVersionGeneralization},~\cite[Chapter 3]{thesis}). When $f$ is \emph{not} an indicator function, one cannot typically compute the dual explicitly and has to upper-bound it leveraging properties of $\mu$ and $f$. In this work, to provide such an upper bound on the dual, we will make use of Markov's inequality. This takes us back to indicator functions for which we can completely characterise the dual. For technical details see~\Cref{app:proofDiv}. This pattern is fundamental whenever one is trying to relate (via upper or lower bounds) the expected value of a function to some divergence/entropy (see~\cite[Chapter 2]{thesis}). 
\end{remark}
The other main tool that will be utilised is H\"older's inequality. In particular, there is a connection between R\'enyi's $\alpha$-information measure and $L^\alpha$-norms of the Radon-Nikodym derivative.
The results we are about to provide are all a consequence of one or multiple applications of H\"older's inequality\footnote{ H\"older's inequality can itself be seen as functional inequality stemming from~\Cref{eq:dualNorm}, considering the usual pairing and, as a starting norm, the regular $L^\alpha$-norm but also as a consequence of~\Cref{eq:legendreFenchelGeneral} with $\psi(x)=\frac{|x|^\alpha}{\alpha}$. For details see~\cite[Section 1.3.1]{thesis}.}. 
\begin{theorem}[\cite{thesis}]\label{alphaExpBound}
	Let $(\X\times\Y,\F,\Pm_{XY}),(\X\times\Y,\F,\Pm_X\Pm_Y)$ be two probability spaces, and assume that $\Pm_{XY}\ll\Pm_X\Pm_Y$. Given an $\F$-measurable function $f:\X\times\Y\to\mathbb{R}^+$, then,
	\begin{align}
		\Pm_{XY}(f) &\leq \left\lVert\left\lVert f \right\rVert_{L^\beta(\Pm_X)}\right\rVert_{L^{\beta'}(\Pm_X)} \cdot  \left\lVert\left\lVert\frac{d\Pm_{XY}}{d\Pm_X\Pm_Y}\right\rVert_{L^\alpha(\Pm_X)}\right\rVert_{L^{\alpha'}(\Pm_Y)} \label{eq:alphaExpBound}
	\end{align}
	where $\beta,\alpha,\beta',\alpha'$ are such that $1=\frac1\alpha+\frac1\beta = \frac{1}{\alpha'}+\frac{1}{\beta'}$. Given a measurable function $g$, $\left\lVert g \right\rVert_{L^\alpha(\mu)}$ denotes the $\alpha$-Norm of $g$ under $\mu$ \textit{i.e.}, $\left(\int g^\alpha d\mu\right)^\frac{1}{\alpha}.$
\end{theorem} 
\begin{remark}
	~\Cref{alphaExpBound} (and corresponding generalisations including Orlicz and Amemiya norms) has already  appeared in~\cite{fullVersionGeneralization} in a slightly less general form, and in~\cite{thesis} in a variety of forms. It  has been re-stated here for ease of reference. 
\end{remark}\noindent~\Cref{alphaExpBound} provides multiple degrees of freedom:
	\begin{itemize}
		\item the parameters characterising the norms: $\alpha,\alpha'$;
		\item the (positive-valued) function $f$.
	\end{itemize}
Three choices of the above are meaningful to us:
\begin{enumerate}
	\item $\alpha'=\alpha$, which makes R\'enyi's Divergence of order $\alpha$ appear on the right-hand side of~\Cref{eq:alphaExpBound} (as a norm of the Radon-Nikodym derivative);\label{en:choice1}
	\item $\alpha'\to 1$, which makes Sibson's Mutual Information of order $\alpha$ appear on the right-hand side of~\Cref{eq:alphaExpBound};\label{en:choice2}
	\item $f=\mathbbm{1}_E$, which allows us to relate the probability of the same event under the joint and a function of the product of the marginals (and an information measure). \label{en:choice3}
\end{enumerate} The choices described in~\Cref{en:choice1} and~\Cref{en:choice3} give rise to the following corollary:
\begin{corollary}[{\cite[Corollary 6]{fullVersionGeneralization}}]
\label{alphaDivBound}
	Given $E\in\F$ and $\alpha > 1$, we have that:
	\begin{equation}
		\Pm_{XY}(E)\leq (\Pm_X\Pm_Y(E))^{\frac{\alpha-1}{\alpha}}\exp\left(\frac{\alpha-1}{\alpha}D_\alpha(\Pm_{XY}\|\Pm_X\Pm_Y)\right),
	\end{equation}
\end{corollary}
\noindent while the choices described in~\Cref{en:choice2} and~\Cref{en:choice3} give rise to the following corollary:
\begin{corollary}[{\cite[Corollary 1]{fullVersionGeneralization}}]
\label{sibsMIBoundCor}
	Given $E\in\F$, we have that:
	\begin{align}
		\Pm_{XY}(E) &\leq\left(\esssup_{\Pm_y} \Pm_X(E_Y)\right)^{1/\beta} \cdot\Pm_Y\left(\Pm_X^{1/\alpha} \left( \left(\frac{d\Pm_{XY}}{d\Pm_X\Pm_Y}\right)^\alpha\right)\right) \label{sibsNonVerdu}\\ &=\left(\esssup_{\Pm_y} \Pm_X(E_Y)\right)^{1/\beta}\cdot \exp\left(\frac{\alpha-1}{\alpha} I_{\alpha}(X,Y)\right), \label{sibMIBound}
	\end{align}
	where $I_\alpha(X,Y)$ is the Sibson mutual information of order $\alpha$,~\citep{verduAlpha}. Moreover, $\alpha$ and $\beta$ are such that $\frac{1}{\alpha} + \frac{1}{\beta} =1$.
\end{corollary}

	\section{Main Results: Lower Bounds on the Risk}\label{sec:riskLBThroughProbs}
	The very first result one can provide stems from an immediate application of~\Cref{sibsMIBoundCor} in conjunction with Markov's Inequality:
	\begin{theorem}\label{thm:sibsMIResultBayesRisk}
		Consider the Bayesian framework described in~\Cref{sec:bayesianFramework}. The following must hold for every $\alpha>1$ and $\rho>0$:
		\begin{equation}
			R_B\geq \rho\left(1- \exp\left(\frac{\alpha-1}{\alpha}\left(I_\alpha(W,X) + \log(L_W(\rho))\right) \right)\right). \label{eq:sibsMILowerBound}
    \end{equation}
    Moreover, taking the limit of $\alpha\to \infty$ one recovers the following:
    \begin{equation}
			R_B\geq \sup_{\rho>0}\rho\left(1- \exp\left(\ml{W}{X} + \log(L_W(\rho)) \right)\right). \label{eq:maximalLeakgeResultBayesRisk}
		\end{equation}
\end{theorem}
The proof can be found in~\Cref{app:proofIalpha}.
	Two remarks are in order:
	\begin{itemize}
		\item It is important to notice that the behaviour of~\Cref{eq:sibsMILowerBound} is fundamentally different from~\cite[Theorem 1]{bayesRiskRaginsky}. In~\cite[Theorem 1]{bayesRiskRaginsky} the dependence is linear with respect to the Mutual Information and logarithmic in $L_W(\rho)$ while in~\Cref{thm:sibsMIResultBayesRisk} there is an exponential dependence on $I_\alpha$ and linear in $L_W(\rho).$ 
		\item \Cref{thm:sibsMIResultBayesRisk} introduces a new parameter $\alpha>1$ to optimise over. The presence of $\alpha$ leads to a trade-off between the two quantities for a given $\rho$, $I_\alpha(W,X)$ and $L_W(\rho)$: $\frac{\alpha-1}{\alpha}I_\alpha(W,X)$ will increase with $\alpha$ whereas $L_W(\rho)^{\frac{\alpha-1}{\alpha}}$ will decrease with $\alpha$. 
	\end{itemize}
    An interesting characteristic of~\Cref{eq:maximalLeakgeResultBayesRisk} is that $\ml{W}{X}$ depends on $W$ only through the support. This allows us to provide, essentially for free, an even more general lower-bound on the risk. Indeed, ignoring $L_W(\rho)$ for a moment, %(it can be upper-bounded regardless of $\Pm_W$ as it is done in Sec. \ref{sec:ex1}), 
	for a fixed family of $\mathcal{P}_{X|W}$, $\ml{W}{X}$ has the same value regardless of $\mathcal{P}_W$ (as long as the support of $W$ remains the same). 
	We can also walk the same path undertaken in~\cite{fullVersionGeneralization} and derive a variety of lower bounds involving a variety of information measures.
	\begin{theorem}
		\label{thm:fDivBoundBayesRisk}
		Consider the Bayesian framework described in~\Cref{sec:bayesianFramework}. Let $\varphi:[0, +\infty) \to \mathbb{R}$ be a monotone, strictly convex function and suppose that the generalised inverse, defined as $\varphi^{-1}(y) = \inf\{t\geq 0 : \varphi(t) > y\}$, exists. Then for every $\rho>0$ and every estimator $\hat{W}$, if $\varphi$ is non-decreasing one has the following
		\begin{equation}
			\Pm_{W\hat{W}}(\ell(W, \hat{W})) \geq \rho\left(1- L_W(\hat{W},\rho)\cdot \varphi^{-1}\left(\frac{I_{\varphi}(W, \hat{W})+(1-L_W(\hat{W},\rho))\cdot\varphi^\star(0)}{L_W(\hat{W},\rho)}\right)\right),\label{eq:fDivBoundBayesRiskIncr}
		\end{equation}
    while if $\varphi$ is non-increasing one recovers the following:
            \begin{equation}
			\Pm_{W\hat{W}}(\ell(W, \hat{W})) \geq \rho\left(1-L_W(\hat{W},\rho)\right)\cdot \varphi^{-1}\left(\frac{I_{\varphi}(W, \hat{W})+L_W(\hat{W},\rho)\cdot\varphi^\star(0)}{1-L_W(\hat{W},\rho)}\right). \label{eq:fDivBoundBayesRiskDecr}
		\end{equation}
		
	\end{theorem}
 The proof can be found in~\Cref{app:proofDiv}.
 \begin{remark}[Recovering Mutual Information]
 A natural question is whether~\Cref{thm:fDivBoundBayesRisk} also includes Shannon's Mutual Information (and, consequently, the results in~\cite{bayesRiskRaginsky}). Selecting $\varphi(x)=x\log x$ is problematic as the function is non-monotonic and its inverse would not have a closed-form expression one could leverage in~\Crefrange{eq:fDivBoundBayesRiskIncr}{eq:fDivBoundBayesRiskDecr}. However, following the same steps undertaken in~\Cref{app:proofDiv}, with $\varphi(x)=x\log(x)$, but with a different choice of $f=-\tilde{\lambda}\ell-\log \mathcal{P}_{W}\mathcal{P}_{\hat{W}}(\exp(-\tilde{\lambda}\ell))+1$ and $\tilde{\lambda}=-\frac1\rho \log\left(\mathcal{P}_W\mathcal{P}_{\hat{W}}(\{\ell < \rho\})\right)$,~\Cref{eq:varReprJoint} does lead to~\cite[Theorem 1]{bayesRiskRaginsky}.
 \end{remark}
 \begin{remark}[Simplification]
      If $\varphi^\star(0) \leq 0$, the expressions in~\Cref{eq:fDivBoundBayesRiskIncr,eq:fDivBoundBayesRiskDecr} can be respectively simplified as follows, if $\varphi$ is non-decreasing:
		\begin{align}
			\Pm_{W\hat{W}}(\ell(W, \hat{W})) \geq \rho\left(1- L_W(\hat{W},\rho)\cdot \varphi^{-1}\left(\frac{I_{\varphi}(W, \hat{W})}{L_W(\hat{W},\rho)}\right)\right), \label{eq:fDivBoundBayesRiskSimplifiedIncr}
		\end{align}
     while if $\varphi$ is non-increasing:
        \begin{align}
			\Pm_{W\hat{W}}(\ell(W, \hat{W})) \geq \rho\left(1- L_W(\hat{W},\rho)\right)\cdot \varphi^{-1}\left(\frac{I_{\varphi}(W, \hat{W})}{1-L_W(\hat{W},\rho)}\right). \label{eq:fDivBoundBayesRiskSimplifiedDecr}
		\end{align}
  The assumption $\varphi^\star(0)\leq 0$ holds indeed true in a variety of cases (cf.~\Cref{thm:lowerBoundHellinger}).
 \end{remark}

	Although~\Cref{thm:fDivBoundBayesRisk} represents a quite general result, in order to apply it to the Bayesian Risk setting (and provide an estimator-independent lower-bound) one has to carefully select $\varphi$. In particular, one has to render the right-hand side of~\Cref{eq:fDivBoundBayesRiskIncr} (or~\Cref{eq:fDivBoundBayesRiskSimplifiedIncr}) independent of $\hat{W}=\phi(X)$. In order to do that, the following two quantities need to be rendered independent of $\hat{W}$:
	\begin{enumerate}
		\item The information measure (\textit{e.g.}, through the data-processing inequality $I_{\varphi}(W, \hat{W}) \leq I_{\varphi}(W, X)$); \label{ineq1}
		\item The quantity $L_W(\hat{W}, \rho)$ (which can be easily upper-bounded in the following way: \\$L_W(\hat{W}, \rho) \leq \sup_{\hat{w}} L_W(\hat{w}, \rho) = L_W(\rho)$). \label{ineq2}
	\end{enumerate}
	For simplicity, consider~\Cref{eq:fDivBoundBayesRiskSimplifiedIncr} and introduce the following object
	\begin{equation}\label{eq:functional_G}
	G_\varphi(I_\varphi, L_W) = L_W(\hat{W},\rho)\cdot \varphi^{-1}\left(\frac{I_{\varphi}(W, \hat{W})}{L_W(\hat{W},\rho)}\right).
	\end{equation}
	To use the two inequalities just stated in~\Cref{ineq1}) and~\Cref{ineq2}), one needs that for a given choice of $\varphi$, $G_\varphi(I_{\varphi}, L_W)$ is increasing in $I_{\varphi}$ for a given value of $L_W$ and vice versa. This allows us to further lower bound~\Cref{eq:fDivBoundBayesRiskSimplifiedIncr} and render the quantity independent of the specific choice of $\phi$. Analogously, one can state similar assumptions in order to apply the same reasoning to~\Cref{eq:fDivBoundBayesRiskSimplifiedDecr}. Hence, starting from~\Cref{risk} one can provide a lower bound on the risk $R_B$ that is independent of $\phi$.\\
	Let us now look at some specific choices of $\varphi$ such that $G_\varphi$ satisfies the desired properties and, thus, for which a bound on the Bayesian risk can be retrieved. 
	\begin{corollary}\label{thm:lowerBoundHellinger}
		Consider the Bayesian framework described in~\Cref{sec:bayesianFramework}. The following must hold for every $p>1$ and $\rho>0$:
		\begin{equation}
			R_B\geq \rho\left(1- L_W(\rho)^{\frac{p-1}{p}}\cdot  \left((p-1) \mathcal{H}_p(W,X) + 1\right) ^\frac{1}{p} \right).\label{eq:lowerBoundHellinger}
		\end{equation}
	\end{corollary}
	The proof of~\Cref{thm:lowerBoundHellinger} is in~\Cref{app:proofHellinger}.
	Restricting the choice of $\varphi$ to this family of polynomials one can thus state the following lower-bound on the risk:
	\begin{equation}\label{equation:lowerBoundHellinger}
		R_B\geq \sup_{\rho > 0} \sup_{p > 1}\rho\!\left(1- L_W(\rho)^{\frac{p-1}{p}}\cdot  \left((p-1) \mathcal{H}_p(W,\hat{W}) + 1\right) ^\frac{1}{p} \right).
	\end{equation}
	\begin{remark}\label{rmk:lowerBoundRenyiDivergence}
		Using the one-to-one mapping connecting Hellinger divergences and R\'enyi's $\alpha$-Divergence, the bound above can be re-written as follows:
		\begin{align}
			R_B\geq \sup_{\rho >0} \sup_{\alpha>1} \rho\left(1- L_W(\rho)^{\frac{\alpha-1}{\alpha}}\cdot 
			\exp\left(\frac{\alpha-1}{\alpha} D_\alpha(\Pm_{W\hat{W}}\|\Pm_W\Pm_{\hat{W}})\right) \right).
		\end{align}
        Moreover, via this relationship one can see that, for a given $\alpha>1$:
	\begin{align}
		\exp\left(\frac{\alpha-1}{\alpha}I_\alpha(W,X^n)\right)  &=\exp\left(\frac{\alpha-1}{\alpha} \inf_{\Q_{X^n}}D_\alpha(\Pm_{WX^n}\|\Pm_W\Q_{X^n})\right) \\
		&\leq \exp\left(\frac{\alpha-1}{\alpha} D_\alpha(\Pm_{WX^n}\|\Pm_W\Pm_{X^n})\right) \\
		&= \left(\mathcal{H}_\alpha(W,X^n)\right)^\frac{1}{\alpha}.
	\end{align}
 Consequently, given the similarity between~\Cref{eq:lowerBoundHellinger} and~\Cref{eq:sibsMILowerBound} one can easily see that leveraging $I_\alpha$, when possible, provides a larger lower-bound than the Hellinger divergence. However, as we will see, in a variety of settings, the Hellinger $\alpha$-divergence can be computed explicitly while $I_\alpha$ cannot. For this reason, we will often leverage the Hellinger divergence to provide closed-form lower bounds on the Risk.
	\end{remark}
	Clearly, a number of results can be derived from~\Cref{thm:fDivBoundBayesRisk}. Each of them with potentially interesting applications in specific Bayesian Estimation settings. In this work, we will mostly focus on Sibson's $\alpha$-Mutual Information and Hellinger $p$-Divergences. In the spirit of leveraging the generality of~\Cref{thm:fDivBoundBayesRisk}, we also provide a bound involving a novel information measure $E_{\gamma,\zeta}$, strongly inspired by the so-called $E_\gamma$-Divergence~\cite[Equation~(66)]{fDivegerceInequalities},~\cite[Page~2314]{polyianksiyConverse}, also known in the literature as the Hockey-stick Divergence and whose application in this framework has been explored in~\citep{CalmonSDPIHockeyStick}. Its definition is the following:
	\begin{definition}
		Let $(\Omega,\F)$ be a measurable space and let $\mu$ and $\nu$ be two probability measures defined on the space.  Denote with $\varphi_{\gamma,\zeta}(x)= \max\{0,\zeta x-\gamma \}-\max\{0, \zeta-\gamma\}$ with $\zeta> 0$ and $\gamma\geq 0$. The function $\varphi_{\gamma,\zeta}(x)$ is convex, increasing and is such that $\varphi_{\gamma,\zeta}(1)=0$. Assume that $\nu\ll\mu$, then  define the following object:
		\begin{equation}
			E_{\gamma,\zeta}(\nu\|\mu) = D_{\varphi_{\gamma,\zeta}}(\nu\|\mu).
		\end{equation}
		Moreover, whenever $\nu=\Pm_{XY}$ and $\mu=\Pm_X\Pm_Y$  we denote (with a slight abuse of notation) $E_{\gamma,\zeta}(\Pm_{XY}\|\Pm_{XY})$ with $E_{\gamma,\zeta}(X,Y).$
		If $\zeta = 1$ then one recovers the usual $E_\gamma$-divergence.
	\end{definition}
	Leveraging it, one can provide the following result in this framework:
	\begin{corollary}\label{thm:lowerBoundGeneralHockeyStick}
		Consider the Bayesian framework described in~\Cref{sec:bayesianFramework}. The following must hold for every $\zeta>0$, $\gamma\geq 0$, and $\rho>0$:
		\begin{equation}\label{eq:hockeyStickGeneralBound}
			R_B\geq \rho\left(1 - \frac{E_{\gamma, \zeta}(W,\hat{W}) + \gamma L_W(\rho) + \max\{0,\zeta-\gamma\}}{\zeta}\right).
		\end{equation}
	\end{corollary}
	
	One can thus retrieve the following lower-bound on the risk:
	\begin{equation}\label{eq:lowerBoundHockeyStick}
		R_B\geq \sup_{\rho>0}\sup_{\zeta>0, \gamma \geq 0}\rho\left(1 - \frac{E_{\gamma, \zeta}(W,\hat{W}) + \gamma L_{W}(\rho)+ \max\{0,\zeta-\gamma\}}{\zeta}\right).
	\end{equation}
 The proof is in~\Cref{app:proofHockeyStick}.
	\begin{remark}
		Setting $\zeta = 1$ in~\Cref{eq:hockeyStickGeneralBound} one recovers~\cite[Remark 1]{CalmonSDPIHockeyStick}.
		In fact, by introducing an additional degree of freedom through the $\zeta$ parameter in~\Cref{eq:lowerBoundHockeyStick}, the resulting lower-bound can only be tighter than~\cite[Remark 1]{CalmonSDPIHockeyStick}.
	\end{remark}
	Using these results one can provide meaningful lower bounds on the Risk in a variety of settings of interest, as we will see in~\Cref{sec:bayesRiskExample}.
    We will now explore a few natural extensions over the framework just presented that follow from either a slight change of perspective or from a slight alteration of the observation model considered. \subsection{Leveraging SDPIs}\label{sec:estimationNoise}
    A key step in the results proved here (as well as in~\citep{bayesRiskRaginsky}) consists of leveraging the Markov Chain $W-X-\hat{W}$ along with DPI as follows: $I_\varphi(W,\hat{W})\leq I_\varphi(W,X^n)$. In case more information is available with respect to the kernel linking $X^n$ and $\hat{W}$ (\textit{i.e.}, more information on how the estimation is executed) then one can leverage the corresponding SDPI inequality as follows:$$I_\varphi(W,\hat{W})\leq I_\varphi(W,X^n)\eta_\varphi(\Pm_{X^n},\Pm_{\hat{W}|X^n}) \leq I_\varphi(W,X^n)\eta_\varphi(\Pm_{\hat{W}|X^n}),
	$$
    potentially providing a refinement over the results that can be advanced. Moreover,
    in a variety of settings one does not have direct access to samples (but rather access to noisy copies or privatised versions). In this case, one can provide refinements over the bounds via SDPI's and, consequently, provide results tailored to the type of noise that has been injected. %\textcolor{red}{Question: Can we say something on the SDPI of $E_{\gamma,\zeta}$? }
    Consider the following setting in which the samples $X^n$ generated from $W$ are not directly observed but rather a sequence $Z^n$, where $Z^n$ is a noisy/privatised version of $X^n$ obtained through the sequence of Markov  Kernel $K_1,\ldots,K_n$ where for every $i\geq 1$, $K_i: \mathcal{F}\times\mathcal{X}\to [0,1]$ and $\Pm_{Z_i} = \Pm_{X_i}K_{i}$. The goal is to estimate $W$ from the sequence $Z^n$ via an estimator $\psi:\Z^n\to\W$. Given a loss function $\ell:\mathcal{W}\times\hat{\mathcal{W}}\to \mathbb{R}^+$, the noisy Bayesian Risk is thus
    defined as
	\begin{equation}
     R_B^{\text{noisy}} = \inf_{\psi} \Pm_{WZ}(\ell(W, \psi(Z^n)))=\inf_{\psi} \Pm_{W\hat{W}}(\ell(W, \hat{W})),
	\end{equation}
    and it can be lower-bounded similarly to the non-private/noisy case via~\Cref{thm:fDivBoundBayesRisk}. For simplicity of exposition, we will consider a single kernel $K$ and, consequently, one has that $Z^n$ is obtained from $X^n$ through the tensor-product $K^{\otimes n}$. In this case, one has the Markov chain $W-X^n-Z^n-\hat{W}$ and can leverage SDPI twice as follows:
    \begin{equation}
        I_\varphi(W,\hat{W}) \leq \eta_{\varphi}(\Pm_{W|Z^n})I_\varphi(W,Z^n) \leq \eta_{\varphi}(\Pm_{W|Z^n})\eta_{\varphi}(\Pm_{X^n},K^{\otimes n})I_\varphi(W,X^n),
    \end{equation}
     and consequently state the following result
    \begin{corollary}\label{thm:noisyEstimation}
    Consider the private Bayesian framework considered above. Denote with $Z^n$ the private samples obtained from $X^n$ through the tensor product of a kernel $K$.
    Let $\varphi:[0, +\infty) \to \mathbb{R}$ be a monotone convex function and suppose that the generalised inverse, defined as $\varphi^{-1}(y) = \inf\{t\geq 0 : \varphi(t) > y\}$, exists. Assume as well that the function $G_{\varphi}$ defined in~\Cref{eq:functional_G} is non-decreasing in both arguments. Then, for every estimator $\hat{W}$, if $\varphi$ is non-decreasing one has the following
		\begin{equation}
			\Pm_{W\hat{W}}(\ell(W,\hat{W})) \geq \sup_{\rho>0} \rho \left(1- L_W(\hat{W},\rho)\cdot \varphi^{-1}\left(\frac{\eta_{\varphi}(\Pm_{\hat{W}|Z^n})\eta_{\varphi}(\Pm_{X^n},K^{\otimes n})I_{\varphi}(W, X^n)+(1-L_W(\rho))\varphi^\star(0)}{L_W(\rho)}\right)\right),\label{eq:fDivBoundBayesRiskIncrLDP}
		\end{equation}
    whereas if $\varphi$ is non-increasing one recovers the following:
            \begin{equation}
			\Pm_{W\hat{W}}(\ell(W,\hat{W})) \geq \sup_{\rho>0} \rho \left(1-L_W(\hat{W},\rho)\right)\cdot \varphi^{-1}\left(\frac{\eta_\varphi(\Pm_{\hat{W}|Z^n})\eta_{\varphi}(\Pm_{X^n},K^{\otimes n})I_{\varphi}(W, X^n)+L_W(\rho)\cdot\varphi^\star(0)}{1-L_W(\rho)}\right). \label{eq:fDivBoundBayesRiskDecrLDP}
		\end{equation}
	
    \end{corollary}
    \begin{remark}
        Notice that in case one does \emph{not} estimate from noisy observations $Z^n$ then one can still consider the same setting as in~\Cref{thm:noisyEstimation} with $K$ representing the kernel associated to the identity mapping \textit{i.e.}, $K(y|x) = \delta_{x}(y)$. In this case one has that $\Pm_{X^n}=\Pm_{Z^n}$, $\eta_{\varphi}(\Pm_{X^n},K^{\otimes n})=1$ and, consequently, $\eta_\varphi(\Pm_{\hat{W}|Z^n})=\eta_\varphi(\Pm_{\hat{W}|X^n}) $. Hence,~\Cref{thm:noisyEstimation} boils down to a refinement of~\Cref{thm:fDivBoundBayesRisk} without any additional noise injection.
    \end{remark}
    \begin{remark}
        In case one is considering $E_{\gamma,\zeta}$ with $\zeta=1$, then the corresponding contraction parameter has been analysed in~\cite{CalmonSDPIHockeyStick}, where contraction has been shown to be equivalent to Local Differential-Privacy (LDP). \textit{I.e.}, a kernel $K$ is said to be $(\epsilon,\delta)$ ``Locally-Differentially private'' (LDP), if:
    \begin{equation}
        \max_{E\in\mathcal{F},x,\hat{x}} |K(E|x)-e^\epsilon K(E|\hat{x})| \leq \delta.
    \end{equation}
    In this case, one has that $\eta_{\varphi_{e^\epsilon}}(K) \leq \delta $. Moreover, due to~\cite[Lemma 2]{CalmonSDPIHockeyStick} one has that if $K$ is $(\epsilon,\delta)$-LDP then for every $\varphi$ on has that $\eta_{\varphi}(K) \leq (1-(1-\delta)e^{-\epsilon}).$ It is unclear, however, whether there are settings in which said upper-bound is more convenient than others. Indeed, one also has that for every convex function $\varphi$ 
    \begin{equation}
        \eta_{\varphi}(K) \leq \eta_{TV}(K) < (1-(1-\delta)e^{-\epsilon}),
    \end{equation}
    and, for many channels, $\eta_{TV}(K)$ is relatively easy to compute. Moreover, even $\eta_{TV}(K)$ tends to be quite larger than the effective contraction coefficient of the divergence at hand: \textit{e.g.}, if $K=\text{BSC}(\lambda_\epsilon)$ with $\lambda_\epsilon = \frac{1}{1+e^{\epsilon}}$ then $K$ is $(\epsilon,0)$-LDP (cf.~\cite[Example 1]{CalmonSDPIHockeyStick}) and if $\varphi$ is operator-convex (\textit{e.g.}, $\varphi(x)=x\log x$ or $\frac{x^p-1}{p-1} $with $1< p \leq 2$, etc.) then~\cite[Corollary 3.1]{sdpiRaginsky}:\begin{equation}
        \eta_{\varphi}(K) = \left(1-\frac{2}{1+e^{\epsilon}}\right)^2 \ll \left|1-\frac{2}{1+e^{\epsilon}}\right| < (1-e^{-\epsilon}).
    \end{equation}
    \end{remark}
An important feature of~\Cref{thm:noisyEstimation} is that for a subclass of functions $\varphi$ one can leverage tensorisation properties of $\eta_\varphi$. Indeed, if $\varphi$ satisfies the conditions of~\cite[Theorem 3.9]{sdpiRaginsky} then one has that
\begin{equation}
    \eta_\varphi(\mu^{\otimes n},K^{\otimes n}) = \eta_\varphi(\mu,K).\label{eq:tensorisationMax}
\end{equation}
Moreover, if $\varphi$ is operator convex, then one can say the following:
\begin{equation}
    \eta_\varphi(K^{\otimes n}) = \eta_{\chi^2} (K^{\otimes n})\leq 1-(1-\eta_\varphi(K))^n.\label{eq:tensorisationPower}
\end{equation}
Both the results are true for instance, for $\varphi(x) = (x^p-1)/(p-1)$ with $1\leq p \leq 2$ (but the assumptions of~\cite[Theorem 3.9]{sdpiRaginsky} are violated if $p>2$). This means that one can leverage~\Crefrange{eq:tensorisationMax}{eq:tensorisationPower} for the Hellinger divergence $\mathcal{H}_p$ with $1\leq p \leq 2$.
\subsection{Conditioning}\label{sec:conditioning}
	Following the approach undertaken in~\citep{bayesRiskRaginsky}, it is also possible to propose a conditional version of the theorems proposed above, in order to retrieve tighter bounds. For this to happen one needs a definition of conditional information measures. For $\varphi$--Divergences the choice would typically fall on objects of the following form
	\begin{equation}
		I_\varphi(X,Y|Z)= D_\varphi(\Pm_{XYZ}\|\Pm_Z \Pm_{X|Z}\Pm_{Y|Z}).
	\end{equation}
	As for Sibson's $I_\alpha$, the matter becomes slightly more complicated as $I_\alpha(X,Y)=\min_{\Q_Y}D_\alpha(\Pm_{XY}\|\Pm_X\Q_Y)$. In the case of three random variables, it is unclear which factorisation of the joint and which minimisation to consider. Indeed, it has been shown in~\citep{conditionalSibsMI} that several definitions of conditional $I_\alpha$ can be proposed, depending on the operational meaning and corresponding probability bound one needs. In this subsection, we will consider the following conditional version of $I_\alpha$:
	\begin{equation}
		I^{Y|Z}_\alpha(X,Y|Z) = \min_{\Q_{Y|Z}} D_\alpha(\Pm_{XYZ}\|\Pm_{X|Z}\Q_{Y|Z}\Pm_Z).\label{eq:condSMI}
	\end{equation}
	The choice of this specific definition is necessary in order to provide a conditional version of~\Cref{thm:sibsMIResultBayesRisk} and~\Cref{eq:maximalLeakgeResultBayesRisk} similar to~\cite[Theorem 1, Eq. (5)]{bayesRiskRaginsky}.
	Leveraging said definition and the fact that: $$I^{Y|Z}_\alpha(X,Y|Z)\xrightarrow[]{\alpha\to\infty} \ml{X}{Y|Z}$$  one can thus give a conditional version of~\Cref{thm:sibsMIResultBayesRisk} and~\Cref{eq:maximalLeakgeResultBayesRisk}, introducing the following notion of conditional small-ball probability,
	$
		L_{W|U}(U,\rho) = \sup_{\hat{w}\in\hat{\W}} \Pm_{W|U}(\ell(W,\hat{w})< \rho)$:
	\begin{theorem}
		Consider the Bayesian framework described in~\Cref{sec:bayesianFramework}, 
		\begin{equation}
			R_B\geq \sup_{\mathcal{P}_{U|W,X}}\sup_{\rho>0, \alpha\geq1}\rho\left(1- \exp\left(\frac{\alpha-1}{\alpha}(I_\alpha(W,X|U) + \log(\Pm_U(L_{W|U}(U,\rho)))\right)\right),
		\end{equation}
		
		Moreover, taking the limit of $\alpha\to\infty$ one has:
		\begin{equation}
			R_B\geq \sup_{\mathcal{P}_{U|W,X}}\sup_{\rho>0}\rho\left(1-\exp\left(\ml{W}{X|U}+ \log(\Pm_U(L_{W|U}(U,\rho)))\right)
			\right). 
		\end{equation}
	\end{theorem}
 \begin{proof}
     For the selected choice of conditional Sibson Mutual Information (see~\Cref{eq:condSMI}) one has that 
     
\begin{equation}
    I_\alpha(W,\hat{W}|U) = \frac{\alpha}{\alpha-1} \log \left\lVert \left\lVert \left\lVert\frac{d\Pm_{W,\hat{W},U}}{d\Pm_U\Pm_{\hat{W}|U}\Pm_{W|U}}\right\rVert_{L^\alpha(\Pm_{W|U})}\right\rVert_{L^1(\Pm_{\hat{W}|U})}\right\rVert_{L^\alpha(\Pm_{U})}. 
    \label{eq:nestedNormIalpha}
\end{equation}
Consequently, one can prove via H\"older's inequality a result analogous to~\Cref{alphaExpBound} (cf.~\cite[Theorem 17]{thesis}) which implies then,  selecting $f=\mathbbm{1}_{\{\ell(W,\hat{W})<\rho\}}$,  the following for every $\rho>0,\alpha\geq 1$ and every $U$ 
\begin{align}
      \Pm_{W\hat{W}}(\ell(W,\hat{W})<\rho) &= 
     \Pm_{W\hat{W}U}(\ell(W,\hat{W})<\rho)\\ &\leq \Pm_U^{\frac{\alpha-1} {\alpha}}\left(\esssup_{\Pm_{\hat{W}|U}}\Pm_{W|U}(\ell(W,\hat{W})\leq \rho)\right) \cdot\exp\left(\frac{\alpha-1}{\alpha}I_\alpha(W,\hat{W}|U)\right).\label{eq:boundCondIalpha1} \\
     &\leq \Pm_U^{\frac{\alpha-1}{\alpha}}\left(L_{W|U}(U,\rho)\right) \cdot\exp\left(\frac{\alpha-1}{\alpha}I_\alpha(W,X|U)\right) \\
     &= \exp\left(\frac{\alpha-1}{\alpha} \left(I_\alpha(W,X|U) + \log(\Pm_U(L_{W|U}(U,\rho)))\right)\right)
\end{align}
The statement of the theorem then follows from the same sequence of steps (involving Markov's inequality) that led to~\Cref{thm:sibsMIResultBayesRisk}. Moreover, starting from~\Cref{eq:nestedNormIalpha} and taking the limit of $\alpha\to+\infty$ one recovers the following:
\begin{equation}
    I_\infty(W,X|U) = \log \esssup_{\Pm_U}\left\lVert \esssup_{\Pm_{W|U}}\frac{d\Pm_{W,X,U}}{d\Pm_U\Pm_{X|U}\Pm_{W|U}}\right\rVert_{L^1(\Pm_{X|U})},
\end{equation}
which can be seen as being equal to $\ml{W}{X|U}$ (cf.~\cite[Section III.E]{leakageLong}).
 \end{proof}
	The main idea behind using conditional Mutual Information, as presented in~\cite{bayesRiskRaginsky}, is that by choosing an appropriate $U$ it is possible to control the growth of $I(W;X|U)$ and obtain tighter bounds in some cases. In particular, consider the sequence of $n$ samples $X^n$. If the family $\Pm= \{ \mathcal{P}_{X|W=w}: w\in\mathcal{W} \}$ is a subset of a finite-dimensional exponential family and $W$ has a
	density supported on a compact subset of $\mathbb{R}^d$, choosing $U$ to be a conditionally independent copy $\hat{X^n}$ of $X^n$ (given $W$) the mutual information $I(W;X^n|\hat{X}^n)$ will converge to a constant as $n$ grows (rather than grow with $n$)~\cite{bayesRiskRaginsky}. This property seems to be specific to Shannon's Mutual Information. In the examples considered below, there does not appear to be a suitable $U$ that tightens the bounds further for the divergences considered. Nonetheless, we state the result as it may be of interest in other settings.
	
    \section{Examples of application}\label{sec:bayesRiskExample}
	In this section, we apply the results presented in the previous section to four estimation settings. The first three are classical settings, while the fourth comes from a distributed estimation setting: 
 \begin{itemize}
     \item estimation of the mean of a Bernoulli random variable with parameter $W$, where $W$ is assumed to be uniform between $(0,1)$;
     \item the same setting as above with the difference that one does not observes the samples $X^n$ directly but rather a noisy/privatised version $Z^n$, where each $Z_i$ is assumed to be the outcome of $X_i$ after being passed through a Binary Symmetric Channel with parameter $\lambda$ (BSC($\lambda$));
     \item estimation of the mean of a Gaussian random variable with different prior distributions over the mean;
     \item identification of the biased random variable in a $d$-dimensional vector in a distributed fashion (cf., the ``Hide-and-seek'' problem advanced in~\cite{nipsHideAndSeek}).
 \end{itemize} 
 The loss function for the first three cases will be the $L^1$-distance while for the fourth one, we will consider the $0-1$ loss. For the first three cases, the maximisation over $\rho$ in the lower bounds is carried out analytically (details in \Cref{app:rho_maximization}).
 
	\subsection{Bernoulli Bias}\label{sec:ex1} 
	\begin{example}\label{ex:bernoulliBias}
		Suppose that $W \sim \mathcal{U}([0,1])$ and that for each $i\in[n]$, $X_i|W=w \sim \text{Ber}(w)$. Also, assume that $\ell(w,\hat{w}) = |w-\hat{w}|$.
	\end{example}
	Using the sample mean estimator \textit{i.e.}, $\hat{W} = \frac1n \sum_{i=1}^n X_i$, one has that (cf.~\cite[Equation (20)]{bayesRiskRaginsky}): \begin{equation}R_B\leq \frac{1}{\sqrt{6n}}.\label{eq:bernoulliBiasUpperBound}\end{equation}
	Let us now lower-bound the risk in this setting. First, we find that
    \begin{equation}
        L_W(\rho)=\sup_{\hat{w}}\Pm_W(|W-\hat{w}|<\rho)=2\rho.
    \end{equation} %In settings where it is difficult to directly estimate $L_W(\rho)$ but only an upper-bound is available it is useful to follow the technique described in the Appendix.
	%Thus, a simple application of \eqref{newLowerBoundML} allows us to get a result tight in $n$. Given that $L_W(\rho) \leq 2\rho = g(\rho)$, we have that $g$ is clearly an increasing function and thus invertible. Moreover, $g^{-1}(k) = k/2.$
    To obtain a lower-bound involving Maximal Leakage, one can see that (details in~\Cref{app:mlBernoulliBias})
    \begin{equation}
        \ml{W}{X^n} \leq \log\left(2+\sqrt{\frac{\pi n}{2}}\right).\label{eq:upperBoundLeakageBernoulliBias}
    \end{equation}
	Substituting~\Cref{eq:upperBoundLeakageBernoulliBias} in~\Cref{eq:maximalLeakgeResultBayesRisk}, along with $L_W(\rho)= 2\rho$, provides us with the following lower-bound on the risk:
	\begin{align}
		R_B&\geq \sup_{\rho>0} \rho\left(1-\exp\left(\ml{W}{X^n}\right)L_W(\rho)\right) \\
		&\geq \sup_{\rho>0} \rho\left(1-\left( 2+ \sqrt{\frac{\pi n}{2}}\right)2\rho\right). \label{eq:mlBernoulliBiasOverRho}
	\end{align}
	The quantity in~\Cref{eq:mlBernoulliBiasOverRho} is a concave function of $\rho$ and thus we can maximise it. In particular, the maximiser is  $\hat{\rho}=\frac{1}{4\left( 2+ \sqrt{\frac{\pi n}{2}}\right)}$ and plugging it in~\Cref{eq:mlBernoulliBiasOverRho} one gets the following:
	\begin{equation}
		R_B\geq \frac{1}{8\left( 2+ \sqrt{\frac{\pi n}{2}}\right)}\label{eq:maximalLeakageBernoulliBias},\end{equation}
	which, for $n$ large enough (\textit{i.e.}, $n\geq 127/\pi \approx 41$), can be further lower-bounded as follows $$R_B\geq \frac{1}{5\sqrt{2\pi n}}.$$
	Surprisingly, Maximal Leakage already offers a lower-bound that matches the upper-bound up to a constant (cf.~\Cref{eq:bernoulliBiasUpperBound}) without any extra machinery.~\Cref{eq:maximalLeakageBernoulliBias} provides a larger lower-bound than the one provided using Mutual Information (cf.~\cite[Corollary 2]{bayesRiskRaginsky}) for $n\geq 1$. Moreover, the proof in \citep{bayesRiskRaginsky} needs a more complicated setting involving a conditioning with respect to an independent copy of $X^n$ and can only provide an \textit{asymptotic} lower bound on the risk of $1/(16\sqrt{2\pi n})$ (that thus, only holds for $n$ large enough).\\
	On the contrary, given the closed-form expression, Maximal Leakage can be quite easy to compute or upper-bound.
	Moreover, the information measure depends on $\Pm_W$ only through the support. This means that if one has access to an upper-bound on $L_W(\rho)$ that does not employ any knowledge of $\Pm_W$ except for the support (\textit{e.g.}, if $W$ were to be discrete, an upper-bound of $1$ over the probability mass function could suffice) the resulting lower-bound on the risk (in this example), would apply to any $W$ whose support is the interval $[0,1]$. \\
	One can also provide a more general lower-bound involving $I_\alpha$. Indeed, one has that (details in~\Cref{app:sibsonBernoulliBias}), in this setting:
	\begin{equation}
		\exp\left(\frac{\alpha-1}{\alpha}I_\alpha(W,X^n)\right) = \sum_{k=0}^n \binom{n}{k}\left(\frac{\Gamma(k\alpha+1)\Gamma((n-k)\alpha+1)}{\Gamma(n\alpha+2)}\right)^\frac{1}{\alpha}.\label{eq:iAlphaBernoulliBiasValue}
	\end{equation}
	Plugging~\Cref{eq:iAlphaBernoulliBiasValue} in~\Cref{eq:sibsMILowerBound} one obtains the following lower-bound on the risk:
	\begin{equation}
		R_B\geq \sup_{\rho>0}\sup_{\alpha>1}\rho\left(1-(2\rho)^\frac{\alpha-1}{\alpha}\exp\left(\frac{\alpha-1}{\alpha}I_\alpha(W,X^n)\right)\right) \label{eq:iAlphaBernoulliBias}.
	\end{equation}
	The lower-bound in~\Cref{eq:iAlphaBernoulliBias} can clearly only improve the one provided in~\Cref{eq:mlBernoulliBiasOverRho}, as $\ml{W}{X^n}=I_\infty(W,X^n)>I_\alpha(W,X^n)$ for every $\alpha<\infty$. However, differently from~\Cref{eq:mlBernoulliBiasOverRho}, it does not admit a closed-form expression and needs to be computed numerically in order to assess how far it is from the upper-bound. 
	Similarly, one could try to employ a lower-bound that includes Hellinger$-p$ Divergences. The lower-bound on the risk induced by~\Cref{thm:lowerBoundHellinger} is given by 
	\begin{equation}
		R_B\geq \sup_{\rho > 0} \sup_{p > 1}\rho\left(1- (2\rho)^{\frac{p-1}{p}}\cdot  \left(\mathcal{H}_p(W,X^n)\right)^\frac{1}{p} \right).\label{eq:hellingerPBernoulliBias}
	\end{equation}
	The expressions represented in~\Cref{eq:hellingerPBernoulliBias} and~\Cref{eq:iAlphaBernoulliBias} are extremely similar. 
	Thus, one can easily argue that~\Cref{eq:iAlphaBernoulliBias} will always provide a larger lower-bound than~\Cref{eq:hellingerPBernoulliBias} and indeed this is confirmed in the simulations. It is also true that, for some values of $p$, one can actually provide nice closed-form expressions for the lower-bound provided by~\Cref{eq:hellingerPBernoulliBias}. Indeed, in general, one has that (details in~\Cref{app:hellBernoulliBias}): 
	\begin{equation}
		((p-1)\mathcal{H}_p(W,X^n)+1) = 
	(n+1)^{p-1}\sum_{k=0}^n \binom{n}{k}^p\frac{\Gamma(kp+1)\Gamma((n-k)p+1)}{\Gamma(np+2)},\label{eq:hellinger_computation}
	\end{equation}
	Then, with $p=2$ one recovers (details in~\Cref{app:hellBernoulliBias}):
	\begin{equation}
		\mathcal{H}_2(W,X^n)+1 =\chi^2(W,X^n) = \frac{n+1}{2n+1}\cdot \frac{4^n}{\binom{2n}{n}}\leq \frac{16\sqrt{\pi n}}{21}. \label{eq:H2Bernoulli}
	\end{equation}
	Hence, specialising Equation \eqref{eq:hellingerPBernoulliBias} to $p=2$ leads us to:
	\begin{equation}
		R_B\geq \sup_{\rho>0}\rho \left(1-\sqrt{2\rho(\chi^2(W,X^n)+1)}\right).
	\end{equation}
	Solving then the maximisation over $\rho$ and using~\Cref{eq:H2Bernoulli} one can conclude that:
	\begin{align}
		R_B&\geq \frac{2}{27}\cdot \frac{1}{\chi^2(W,X^n)+1} \geq \frac{7}{72\sqrt{\pi n}}.\label{eq:bernoulli_lower_bound_closed_form}
	\end{align}
	Notice that~\Cref{eq:bernoulli_lower_bound_closed_form} also matches the upper-bound up to a constant and, similarly to Maximal Leakage, tightens the result in~\cite[Corollary 2]{bayesRiskRaginsky} while not requiring that $n \to \infty$.
	\begin{remark}
		Stirling's approximation yields $(\chi^2(W,X^n)+1) \sim \frac{\sqrt{\pi n}}{2}$ when $n$ is large. This implies that, for $n$ large, one can show that
		$R_B\gtrsim \frac{4}{27\sqrt{\pi n}}$, thus leading to a slight improvement over~\Cref{eq:bernoulli_lower_bound_closed_form}.
	\end{remark}
	To conclude, one can apply the same steps with the $E_{\gamma,\zeta}$-Divergence.
	The lower-bound on the risk  one can retrieve via~\Cref{thm:lowerBoundGeneralHockeyStick} is the following:n this example can thus be expressed as
    	\begin{align}
		R_B&\geq \sup_{\zeta,\gamma}\sup_{\rho>0}\rho\left(1-\frac{\left(E_{\gamma,\zeta}(W, X^n) + 2\rho\gamma\right)}{\zeta}\right)  \label{eq:bernoulliHockeyStickMax}\\
		&= \sup_{\zeta,\gamma} \frac{(\zeta-E_{\gamma,\zeta}(W, X^n))^2}{8\gamma\zeta}.\label{eq:bernoulli_hockey_stick}
	\end{align}
	The lower-bound in~\Cref{eq:bernoulli_hockey_stick} can be empirically seen to be the best among the ones presented so far (thus beating Hellinger, $I_\alpha$ and, consequently, Maximal Leakage and Mutual Information).
	A direct comparison between the bounds provided here and those already present in the literature can be seen in~\Cref{fig:bernoulli_specific_params} and~\Cref{fig:bernoulli_best_params}. The lower bounds are computed as a function of the number of samples $n$, which we consider to be in the range $\{1, \dots, 50\}$. The figure shows that all the divergences we considered in this work provide a larger (and thus, tighter) lower-bound on the Bayesian risk when compared with results that stem from using Shannon's Mutual Information (cf.~\cite[Corollary 2]{bayesRiskRaginsky}). In particular, the lower-bound involving the $E_{\gamma,\zeta}$-Mutual Information represents the largest among the ones we consider. Given the lack of a closed-form expression for $E_{\gamma,\zeta}$ in this example, the quantity in~\Cref{eq:bernoulli_hockey_stick} was computed numerically. Moreover, in order to verify whether the behaviour (and ordering) of the lower bounds in~\Cref{fig:bernoulli_specific_params} was determined by the specific choices of the parameters $p,\gamma,\zeta$ and $\alpha$, in~\Cref{fig:bernoulli_best_params} the lower bounds on the risk have also been numerically optimised over the respective parameters $p,\gamma,\zeta, \alpha$. As~\Cref{fig:bernoulli_best_params} shows, the lower-bound provided by $E_{\gamma,\zeta}$ remains the best. Notice that the lower-bound involving Mutual Information has no parameter to optimise over (other than $\rho$). Maximal Leakage does not provide the best bound, but it possesses the interesting characteristic of depending on $\Pm_W$ only through the support, thus leading to potential applicability in a variety of settings in which $\Pm_W$ is not accessible. In contrast, Mutual Information, the Hellinger Divergence and the $E_{\gamma,\zeta}$-Divergence all require to know $\Pm_W$.
	The lower bounds on the Risk in this Example can thus be summarised as follows:
	\begin{corollary}
		Consider the setting of Example \ref{ex:bernoulliBias} one has that
		\begin{align}
			R_B\geq \max\left\{ \max_{\zeta>0, \gamma\geq 0}\left\{\frac{(\zeta-E_{\gamma,\zeta}(W, X^n))^2}{8\gamma\zeta}\right\}, \max_{\alpha>1} \left\{\left(\frac{(2\alpha -1)}{2\alpha}\exp\left(\frac{\alpha-1}{\alpha}I_\alpha(W,X^n)\right)\right) ^{-\frac{\alpha}{\alpha-1}}\frac{(\alpha-1)}{(2\alpha-1)}\right\}\right\}.
		\end{align}
	\end{corollary}
	\begin{figure}
		\centering
		\subfloat[The picture shows the behaviour of~\Cref{eq:mlBernoulliBiasOverRho},~\Cref{eq:iAlphaBernoulliBias} with $\alpha=2$,~\Cref{eq:bernoulli_lower_bound_closed_form},~\Cref{eq:bernoulli_hockey_stick} with $\gamma=3$ and $\zeta=1.5$ and~{\cite[{{Equation~(19)}}]{bayesRiskRaginsky} }
		as a function of $n$. The values of $E_{3, 1.5}(W, X^n)$ for each $n$ are computed numerically. Solid lines mean that the corresponding lower-bound is the largest.]{ %\includegraphics[width=0.45\textwidth]{images/bernoulli_example_specific_params.pdf}\label{fig:bernoulli_specific_params}} \qquad
        \includegraphics[width=0.465\textwidth]{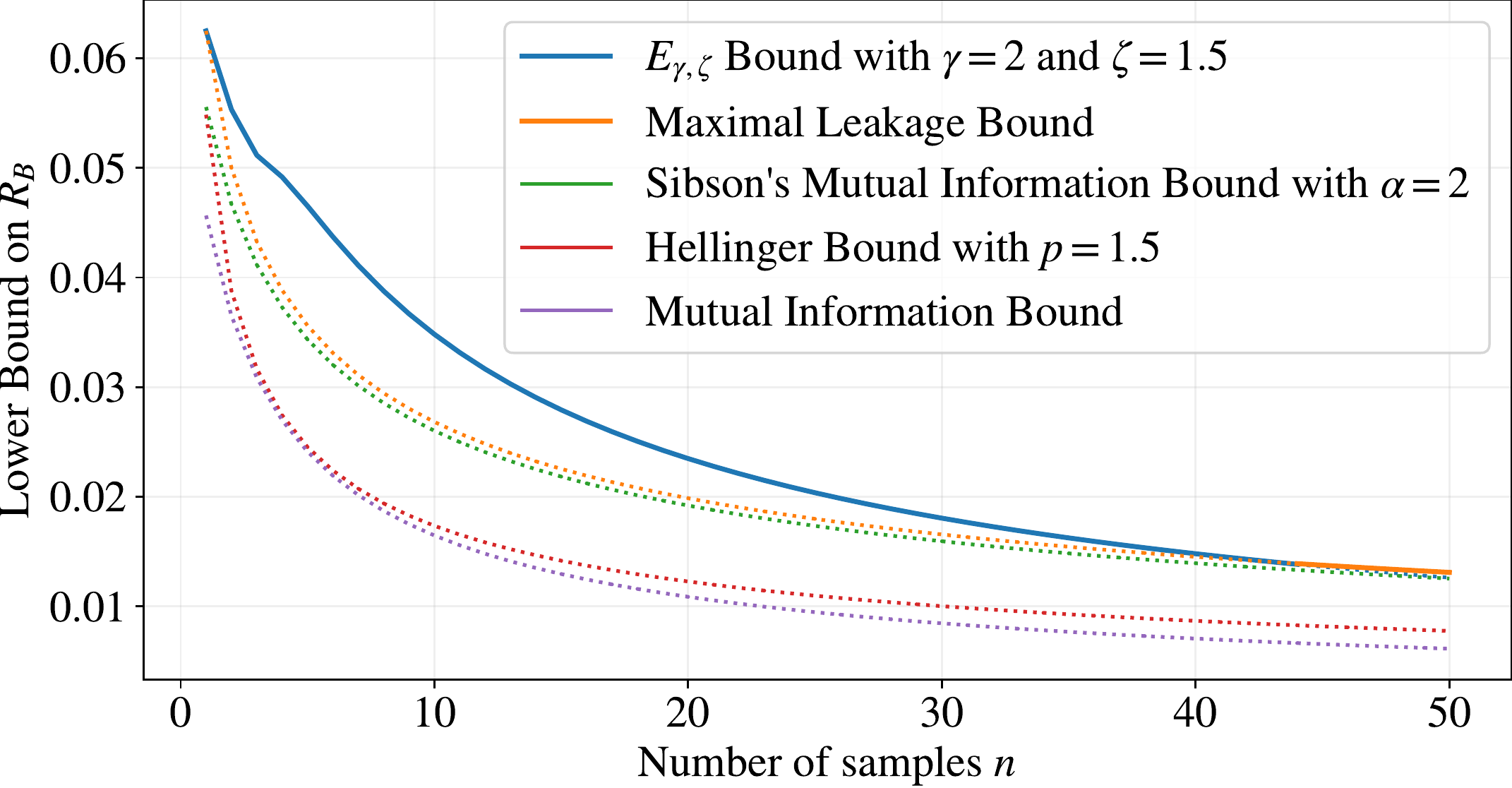}\label{fig:bernoulli_specific_params}} \qquad
		\subfloat[Comparison between the largest lower bounds one can retrieve for different information measures in~\Cref{ex:bernoulliBias}: that is between~\Cref{eq:iAlphaBernoulliBias},~\Cref{eq:hellingerPBernoulliBias},~\Cref{eq:bernoulliHockeyStickMax} and {\cite[Equation~(19)]{bayesRiskRaginsky}}.
		The quantities are analytically maximized over $\rho$ and numerically optimized over, respectively, $\alpha>1,p>1$, $\zeta>0$, and $\gamma\geq 0$. Solid lines mean that the corresponding lower-bound is the largest.]{
            \includegraphics[width=0.465\textwidth]{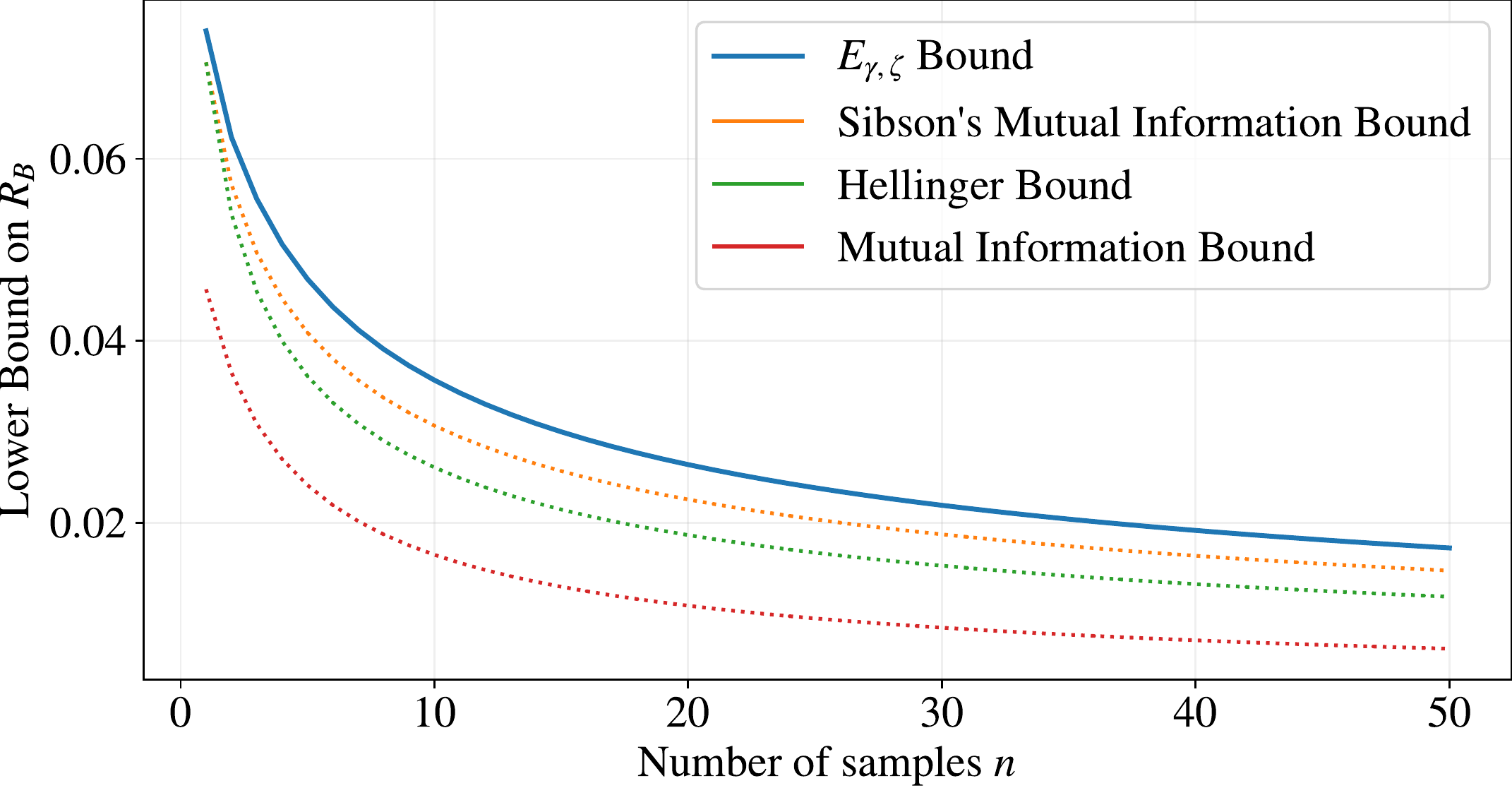}
			\label{fig:bernoulli_best_params}}
   \caption{Comparison of various bounds for~\Cref{ex:bernoulliBias} with and without (numerical) optimisation of parameters.}
	\end{figure}
	
	\subsection{Noisy Bernoulli Bias}\label{sec:noisy_bernoulli_bias}
    Assume, like in~\Cref{sec:ex1}, that $W$ is uniform on the $[0,1]$ interval, $X_i \sim \text{Ber}(W)$. In line with the discussion above, suppose that one observes noisy copies of $X_i$'s denoted with $Z_i$'s, where $Z_i$ is the outcome of $X_i$ after being passed through a Binary Symmetric Channel with parameter $\lambda$ ($K$=BSC($\lambda$)). The purpose is to estimate $W$ through a function of $Z_1,\ldots,Z_n$ \textit{i.e.}, $\hat{W}=\psi(Z^n)$. One thus has the following Markov Chain $W-X^n-Z^n-\hat{W}$. In order to lower-bound the Bayesian risk in this setting, one can use~\Cref{thm:noisyEstimation}. In particular, given the additional injection of noise, it is to be expected that one has a stronger impossibility result with respect to the non-noisy version. This is reflected in the computations below. Let us restrict ourselves to $\varphi$-Divergences as one can then leverage the results in the literature on SDPI constants for the channel considered here. In particular, if one considers Hellinger divergences then the following can be said~\cite[Corollary 3.1]{sdpiRaginsky},~\cite{etaTVKL}:
    \begin{align}
        \eta_p(K) &= (1-2\lambda)^2 \text{ if } 1\leq p\leq 2 \\
        \eta_p(K) &\leq |1-2\lambda| \hspace{0.7em} \text{  if } p>2.
    \end{align}
    Moreover, if $p\leq 2$ then one can leverage tensorisation properties of SDPI-coefficients (see~\Cref{sec:SDPI}) and the following can be said
    \begin{equation}
        \eta_p(\Pm_{X^n},K^{\otimes n}) = \eta_p(\Pm_{X},K) \leq \eta_p(K) = 
        (1-2\lambda)^2 \text{ if } 1\leq p\leq 2 
    \end{equation}
    In this setting, one has that the Hellinger divergence $\mathcal{H}_p(W,X^n)$ is given in~\Cref{eq:hellinger_computation}. With $p=2$ and without making any assumption on $\Pm_{\hat{W}|Z^n}$, one can leverage~\Cref{thm:noisyEstimation} and retrieve the following closed-form expression for the Risk in this setting:
    \begin{align}
        R^{\text{noisy}} &\geq \sup_{\rho>0} \rho\left(1-\sqrt{2\rho((1-2\lambda)^2\chi^2(W,X^n)+1)}\right) \\
        &= \frac{2}{27}\frac{1}{(1-2\lambda)^2\chi^2(W,X^n)+1} \\
        &\geq \frac{2}{27}\frac{1}{(1-2\lambda)^2\frac{16\sqrt{\pi n}}{21}+1} \label{eq:chi2PrivateBernoulli}.
    \end{align}
    Clearly, the denominator in~\Cref{eq:chi2PrivateBernoulli} is smaller than the one in~\Cref{eq:bernoulli_lower_bound_closed_form} thus yielding a larger lower-bound on the risk, this is depicted in~\Cref{fig:hellinger_noisy} for the case $\lambda=0.25$.
    \begin{figure}
		\centering
		\includegraphics[width=.6\linewidth]{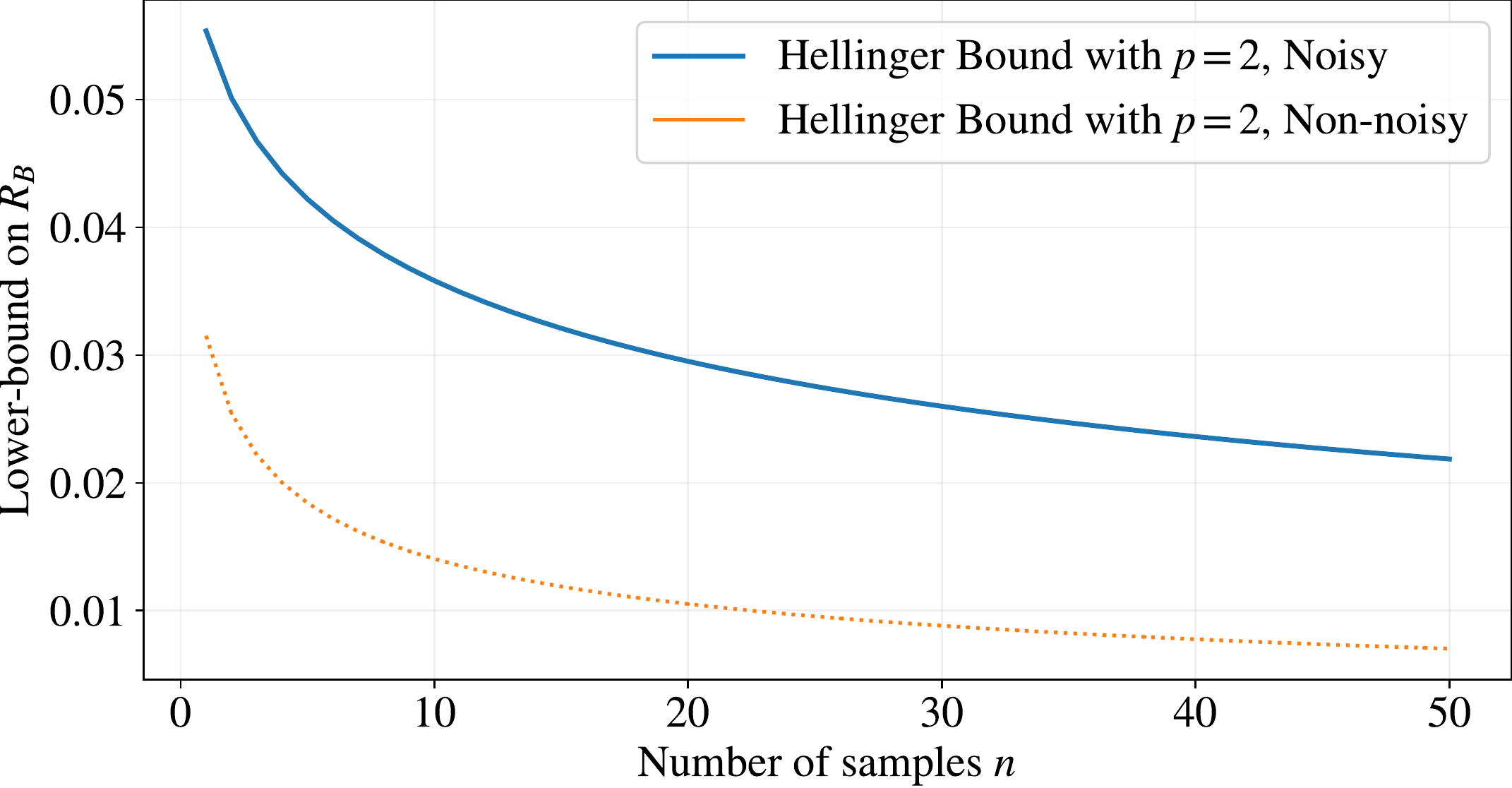}
		\caption{Comparison of the lower bounds in \Cref{eq:bernoulli_lower_bound_closed_form} and \Cref{eq:chi2PrivateBernoulli} for the noisy Bernoulli bias setting described in \Cref{sec:noisy_bernoulli_bias} with $\lambda=0.25$. Solid lines mean that the corresponding lower-bound is the largest.}
		\label{fig:hellinger_noisy}
	\end{figure}
  %Clearly, for every finite $n$, the denominator of~\Cref{eq:chi2PrivateBernoulli} is smaller than that of~\Cref{eq:bernoulli_lower_bound_closed_form}, thus yielding a larger lower-bound on the Risk. With $n$ large enough the two expressions essentially coincide.
  \subsection{Gaussian prior with Gaussian noise (and absolute error)} \label{sec:gaussianExample}
	Another classical and interesting setting is given by the following example:
	\begin{example}\label{gaussianPriorGaussianNoise}
		Assume that $W\sim N(0,\sigma^2_W)$ and that for $i\in[n]$, $X_i = W + Z_i$ where $Z_i \sim N(0,\sigma^2)$. Assume also that the loss is s.t. $\ell(w,\hat{w}) = |w-\hat{w}|$.
	\end{example}
	Using the sample mean estimator one has that: \begin{equation}
		R_B\leq \sqrt{\frac{\sigma^2_W}{1+n\sigma^2_W/\sigma^2}}. \label{eq:gaussianPriorGaussianNoiseUpperBound}
	\end{equation}
	Moreover, given that $\ell(w,\hat{w}) = |w-\hat{w}|$ it is also possible to show that: \begin{equation}
		L_W(\rho) \leq \left(\sup_{w\in\mathbb{R}} \mathcal{P}_W(w)\right)\left(\int_{-\rho}^{\rho} 1du\right) \leq \rho \sqrt{\frac{2}{\sigma^2_W\pi}}.\label{gaussianWSmallBallProbUpperBound}
	\end{equation}
	In this setting, $\ml{W}{X^n}$ is infinite. However, $I_\alpha(W,X^n)$ is finite for every $\alpha<+\infty$. One can thus provide a lower bound on the Risk, resorting to $I_\alpha$ via~\Cref{eq:sibsMILowerBound}. Given that the empirical mean is a sufficient statistic for $W$ in this case, one has that~\cite[Example 5]{verduAlpha}:
	\begin{equation}
		I_\alpha(W, X^n) = I_\alpha\left(W, \frac1n\sum_{i=1}^nX_i\right) = \frac{1}{2}\log\left(1+\alpha n\frac{\sigma^2_W}{
			\sigma^2}\right).
	\end{equation}
	These considerations imply that :
	\begin{align}
		R_B&\geq \sup_{\alpha>1}\sup_{\rho>0} \rho\left(1-\exp\left(\frac{\alpha-1}{\alpha}I_\alpha(W,X^n)\right)\left(\rho \sqrt{\frac{2}{\sigma^2_W\pi}}\right)^\frac{\alpha-1}{\alpha}\right) \\
		&= \sup_{\alpha>1}\sup_{\rho>0} \rho\left(1-\left(\rho \sqrt{\left(1+\alpha n\frac{\sigma^2_W}{
				\sigma^2}\right)\frac{2}{\sigma^2_W\pi}}\right)^\frac{\alpha-1}{\alpha}\right) \\
		&= \sup_{\alpha>1} \frac{1}{(\beta+1)}\left(\frac{\beta}{\beta+1}\right)^\beta \left(\sqrt{\left(1+\alpha n\frac{\sigma^2_W}{
				\sigma^2}\right)\frac{2}{\sigma^2_W\pi}}\right)^{-\frac{1}{\beta}}, \label{eq:gaussianIalpha}
	\end{align}
	remembering that $\beta = \frac{\alpha}{\alpha-1}.$
	Stepping away from Sibson's $\alpha$-Mutual Information one can look at Hellinger $p$-Divergences and $E_{\gamma,\zeta}$ once again. 
	In particular, one has that for $p>1$ (details in~\Cref{app:hellGaussian}):
	\begin{equation}
		\mathcal{H}_p(W,X) = \left(\frac{\left(1+\frac{\sigma_W^2}{\sigma^2}\right)^p}{1 + (2-p)p\frac{\sigma_W^2}{\sigma^2}}\right)^\frac{1}{2}.
	\end{equation}
	Thus, the family of bounds provided by~\Cref{thm:lowerBoundHellinger} can be expressed as follows
	\begin{align}
		R&\geq \sup_{p > 1}\sup_{\rho > 0}\rho\left(1- \left(\frac{2\rho}{\sqrt{2\pi\sigma_W^2}}\right)^{\frac{p-1}{p}}\mathcal{H}^\frac{1}{p}_p(W,X^n)\right) \\
		&= \sup_{p > 1}\sup_{\rho > 0}\rho\left(1- \left(\frac{2\rho}{\sqrt{2\pi\sigma_W^2}}\right)^{\frac{p-1}{p}}\left(\frac{\left(1+\frac{\sigma_W^2}{\sigma^2}\right)^p}{1 + (2-p)p\frac{\sigma_W^2}{\sigma^2}}\right)^\frac{1}{2p}\right)  \\ &= \sup_{p>1} \frac{1}{q+1}\left(\frac{q}{q+1}\right)^q \left(\left(\frac{2}{\sqrt{2\pi\sigma_W^2}}\right)^{\frac{p-1}{p}}\left(\frac{\left(1+\frac{\sigma_W^2}{\sigma^2}\right)^p}{1 + (2-p)p\frac{\sigma_W^2}{\sigma^2}}\right)^\frac{1}{2p}\right)^{-\frac{1}{q}}, \label{eq:gaussianHellingerp}
	\end{align}
	where $q$ represents the H\"older's conjugate with respect to $p$, \textit{i.e.}, $q=\frac{p}{p-1}$.
	%\textcolor{blue}{Just for my reference: $\max_\rho \rho(1-c\rho^\frac1\beta)= \left(\frac{\beta}{c(\beta+1)}\right)^\beta\frac{1}{\beta+1}$ attained at $\rho^\star = \left(\frac{\beta}{(\beta+1)c}\right)^\beta$.}
	
	In particular, setting $p=3/2$ one obtains:
	\begin{equation}
		\mathcal{H}_{3/2}(W, X) =  \sqrt{\frac{\left(1+\frac{\sigma_W^2}{\sigma^2}\right)^\frac32}{1 + \frac{3\sigma_W^2}{4\sigma^2}}},
	\end{equation}
	leading us to a lower bound on the Bayesian risk given by:
	\begin{equation}\label{eq:gaussianHellingerSpecific}
		R_B\geq \frac{81\sqrt{2\pi}}{2048}\sqrt{\frac{\sigma_W^2}{1+n\frac{\sigma_W^2}{\sigma^2}}}.
	\end{equation}
	Similarly to the previous example, one has that~\Cref{eq:gaussianHellingerSpecific} matches the upper-bound up to a constant factor, and provides a strengthening of the bounds obtained in~\cite[Corollary 1]{bayesRiskRaginsky}.
	Repeating the analysis with the $E_{\gamma,\zeta}$-Divergence, one obtains the following:
	\begin{align}
		R_B&\geq \sup_{\rho>0}\rho\left(1-\frac{\left(E_{\gamma,\zeta}(W, X^n) + \frac{2\rho\gamma}{\sqrt{2\sigma^2_W\pi}}\right)}{\zeta}\right)\\
		&= \frac{\sqrt{2\sigma^2_W\pi}\left(\zeta-E_{\gamma,\zeta}(W, X^n)\right)^2}{8\gamma\zeta}.\label{eq:gaussianHockeyStick}
	\end{align}
	%In particular for the case $\zeta=0.75$ and $\gamma=2.2$,~\Cref{eq:gaussianHockeyStick} becomes
	%\begin{align}
	%    R_B &\geq \frac{5\sqrt{2\pi\sigma_W^2}(0.75-E_{2.2, 0.75}(W, X^n))^2}{66}.\label{eq:gaussianHockeyStickSpecific}
	%\end{align}
	\begin{figure}
		\centering
		\subfloat[Setting:~\Cref{gaussianPriorGaussianNoise} with $\sigma_W^2=1$ and $\sigma^2=2$. The picture shows the behaviour of~\Cref{eq:gaussianIalpha} with $\alpha=2$,~\Cref{eq:gaussianHellingerSpecific},~\Cref{eq:gaussianHockeyStick} with $\gamma=2$ and $\zeta=1.5$ and~{\cite[Equation~(16)]{bayesRiskRaginsky}} as a function of $n$. The values of $E_{2,1.5}(W, X^n)$ for each $n$ are computed numerically. Solid lines mean that the corresponding lower-bound is the largest.]{ %\includegraphics[width=0.45\textwidth]{images/gaussian_example_specific_params.pdf}\label{fig:gaussian_specific_params}} \qquad
        \includegraphics[width=0.465\textwidth]{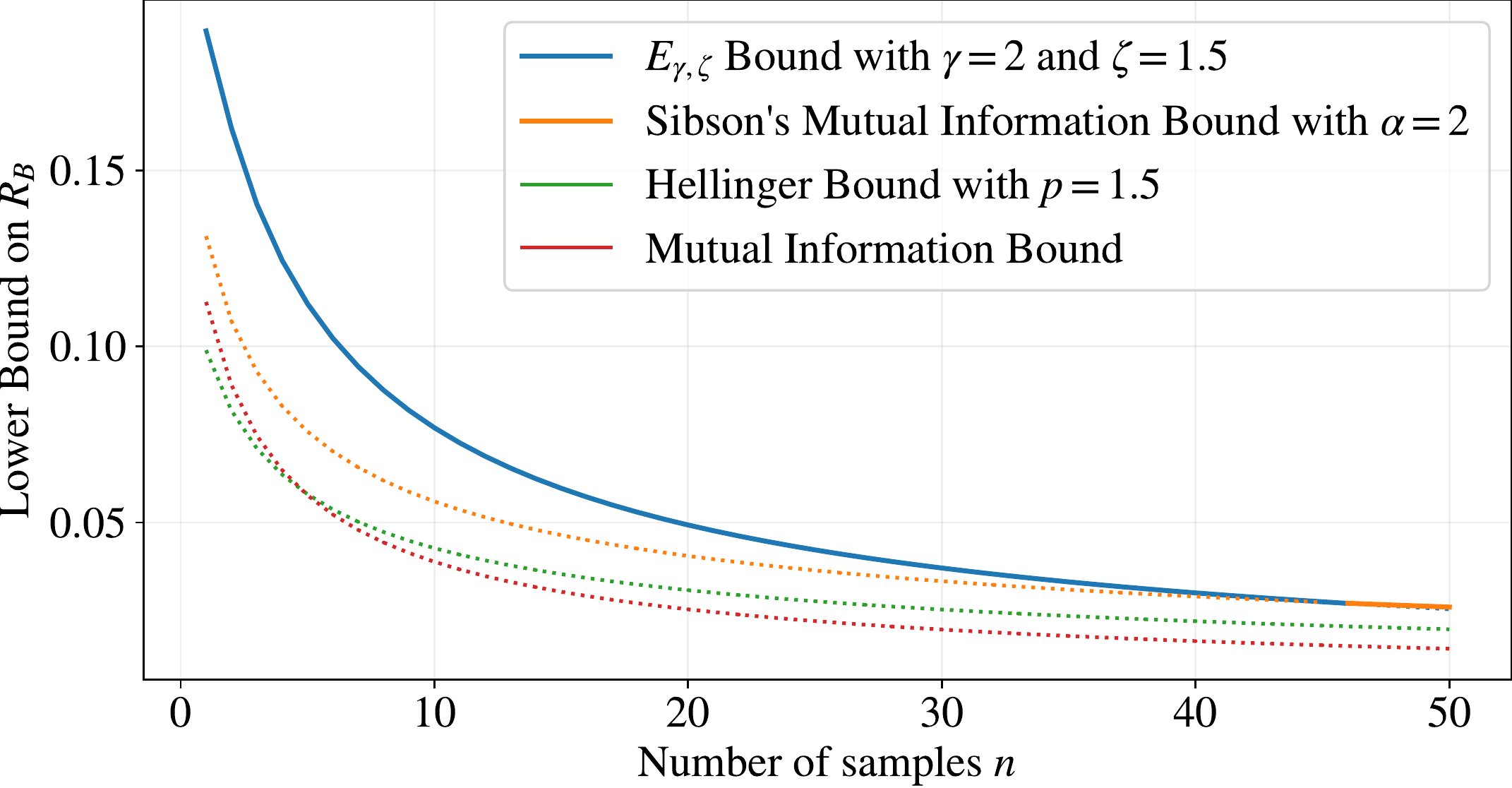}\label{fig:gaussian_specific_params}} \qquad
		\subfloat[Comparison between the largest lower bounds one can retrieve for different information measures in~\Cref{gaussianPriorGaussianNoise}: that is between,~\Cref{eq:gaussianIalpha},~\Cref{eq:gaussianHellingerp},~\Cref{eq:gaussianHockeyStick} with $\zeta=1.5$, and~{\cite[ Equation~(16)]{bayesRiskRaginsky}}.
		The quantities are numerically optimised over, respectively, $\gamma\geq 1$, $p>1$ and $\alpha>1$. The numerical optimisation over the parameter $\zeta$ is not carried out for computational reasons. Solid lines mean that the corresponding lower-bound is the largest.]{
            \includegraphics[width=0.465\textwidth]{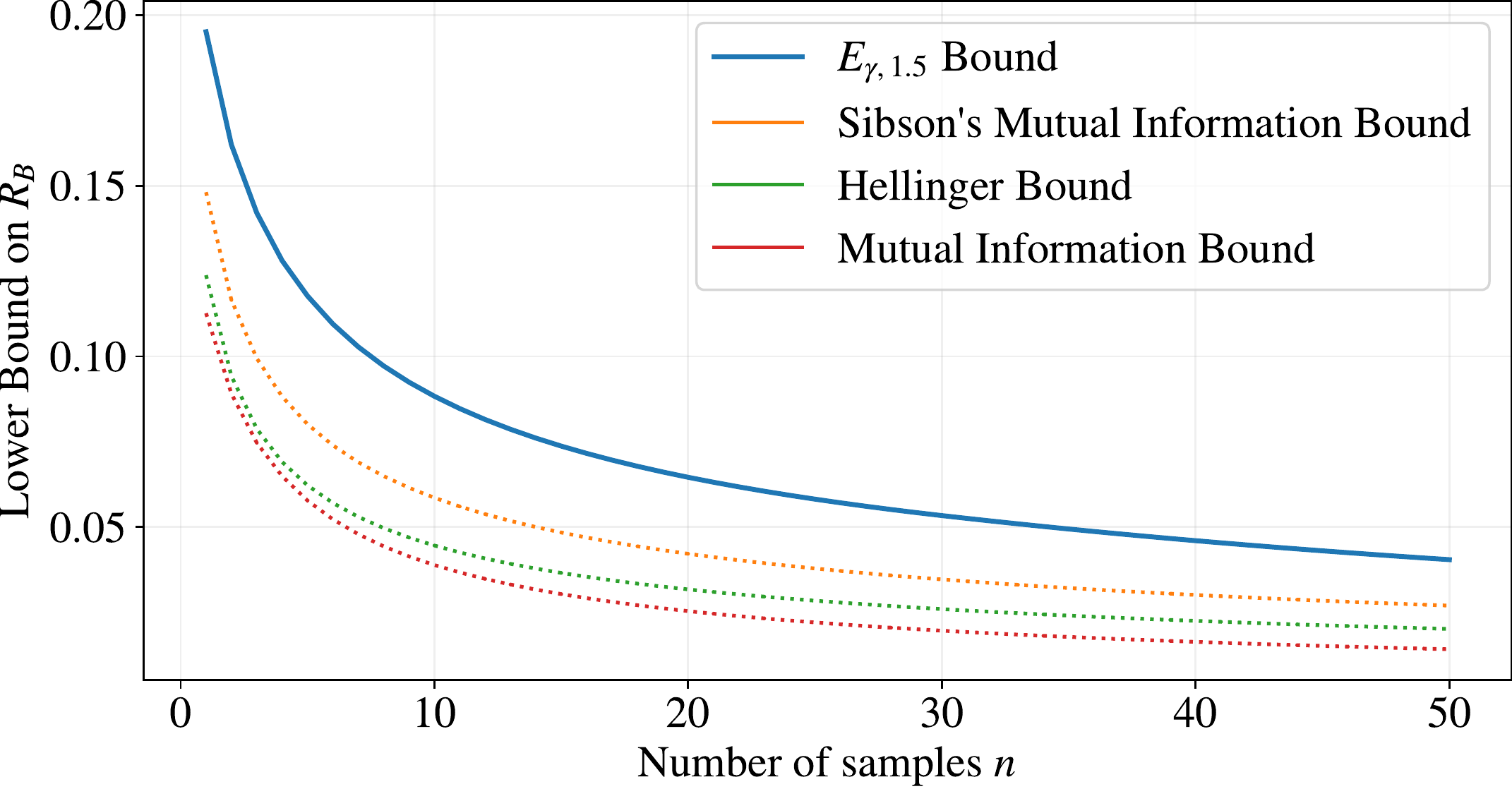}
			\label{fig:gaussian_best_params}}
	\caption{Comparison of various bounds for~\Cref{gaussianPriorGaussianNoise} with and without (numerical) optimisation of parameters.}
	\end{figure}
	Like in~\Cref{ex:bernoulliBias}, one can numerically evaluate~\Cref{eq:gaussianHockeyStick} and compare it with~\Cref{eq:gaussianIalpha},~\Cref{eq:gaussianHellingerSpecific} and \cite[Equation~(16)]{bayesRiskRaginsky}.~\Cref{fig:gaussian_specific_params} and~\Cref{fig:gaussian_best_params} show the resulting lower bounds as a function of the number of samples $n$. One can observe similar behaviors when comparing with the results from previous example: the bounds retrieved through the $\mathcal{H}_p$- and $E_{\gamma,\zeta}$-Divergences are both able to improve on the lower-bound relying on Shannon's Mutual Information. Once again, $E_{\gamma,\zeta}$, (cf.~\Cref{eq:gaussianHockeyStick}) provides the largest lower-bound, while Sibson's $\alpha$-Mutual Information is still able to provide a stronger result than~\Cref{eq:gaussianHellingerp}. Similarly to before, one can also numerically optimise the bounds with respect to the corresponding parameters $\alpha>1,p>1,\zeta>0$ and $\gamma\geq 0$ and the resulting comparison is depicted in~\Cref{fig:gaussian_best_params}.

	\subsection{``Hide-and-seek'' problem}\label{sec:hideAndSeek}
	To conclude, let us consider next a $d$-dimensional distributed estimation problem, known as the ``Hide-and-seek'' problem. It has been first presented in \cite{nipsHideAndSeek} and also studied in \cite{bayesRiskRaginsky}.
	\begin{example}
		Consider a family of distributions $\mathcal{P} = \{\Pm_w : w = 1,\ldots , d\}$ on $\{0, 1\}^d$. Under $\Pm_w$, the
		$w$-th coordinate of the random vector $X \in \{0, 1\}^d$ has bias $\frac{1}{2}+\theta$ while the other coordinates of $X$ are
		independently drawn from Ber$(1/2)$. 
		For $i = 1,\ldots, m$, the $i$-th processor observes $n$ samples $X^n_i$ drawn independently from $\mathcal{P}_W$, and sends a $b$-bits message $Y_i = \varphi(X_i^n,Y^{i-1})$ to the estimator. The estimator
		computes $\hat{W}=\psi(Y^m)$ from the received messages. The risk in this example is defined as: 
		\begin{equation}
			R_M = \inf_{\varphi^m,\psi} \max_{w\in[d]}\mathbb{P}[W\neq \hat{W}].
		\end{equation}
	\end{example}
	A lower-bound for $R_M$ derived in~\cite{nipsHideAndSeek} is as follows:
	\begin{equation}
		R_M \geq 1- \left(\frac3d+5\sqrt{\min\left\{\frac{10\theta nmb}{d},mn\theta^2\right\}} \right) \label{nipsRiskHideAndSeek}
	\end{equation}
	and only holds for $0\leq \theta \leq 1/(4n).$ Additionally, in~\citep{bayesRiskRaginsky} a quite different lower-bound has been proposed:
	\begin{equation}
		R_M \geq 1 - \frac{1}{\log d}\min\left\{\left[1-\left(\frac{1-2\theta}{1+2\theta}\right)^n\right]mb+1, \min(4mn\theta^2,\log d) + 1 \right\}, \label{miRiskHideAndSeek}
	\end{equation}
	and it holds for $0\leq \theta \leq 1/2.$
	Let us now use a na\"ive approach with Maximal Leakage.
	We have that $W-X^{n\times m}-Y^m-\hat{W}$ forms a Markov Chain.
	Thus, $$\ml{W}{\hat{W}} \leq \min\left\{\ml{W}{X^{n\times m}}, \ml{W}{Y^m}\right\}.$$
	We also have that $\ml{W}{Y^m} \leq mb$ and that:
	\begin{align}
		\mathcal{L}(W \to X^{n\times m}) &\leq nm\mathcal{L}(W\to X) \label{eq:hideSeekMLChainRule}\\
		&= nm\log \sum_{x} \max_w \Pm_{X|W=w}(x) \\
		&\leq nm\log \sum_{x} \left(\frac{1}{2}\right)^{d-1} \left(\frac{1}{2}+\theta \right) \\
		&= nm\log (2^d(2^{-d}+ 2^{-d+1}\theta)) \\
		&= nm\log (1+2\theta),
	\end{align}
	Hence:
	\begin{equation}
		\ml{W}{\hat{W}} \leq \min(nm \log (1+2\theta), \log d,  mb).\label{mlHideAndSeek}
	\end{equation}
	Using~\Cref{mlHideAndSeek} in~\Cref{eq:maximalLeakgeResultBayesRisk} we get the following:
	\begin{equation}
		\mathbb{P}(\{\hat{W}\neq W\}) \geq 1 - \frac{\exp(\text{min}\{ mb,\log d, nm\log(1+2\theta)\})}{d}. \label{mlRiskHideAndSeek}
	\end{equation}
	
	%\textcolor{blue}{Lower-bound through $I_{-\infty}$?}
	Notice that~\Cref{mlRiskHideAndSeek} is such that the right-hand side is always greater or equal than $0$. Indeed, assuming $d$ to be fixed and letting $n$ and $m$ grow, we have that the minimum is achieved by $\log d$, and in that case, we have $\mathbb{P}(\{\hat{W}\neq W\}) \geq 0$.
	Here, the difference in behaviour of~\Cref{eq:maximalLeakgeResultBayesRisk} with respect to~\cite[Theorem 1]{bayesRiskRaginsky} is pivotal. Let us now compare the results on a common setting. The setting chosen in~\citep{bayesRiskRaginsky}, where $d=512, b=3d, m=10$ and $\theta = 1/(4n)$ does not represent a choice of parameters where~\Cref{mlRiskHideAndSeek} is interesting. Indeed, for large enough $n$, $nm\log(1+2\theta) = nm\log(1+1/2n) \approx m/2$ and, as a consequence, the expression will converge to a constant determined by the minimum between $mb, \log d, m/2$. Furthermore, both~\Cref{nipsRiskHideAndSeek} and~\Cref{miRiskHideAndSeek} have a term that depends on $mn\theta^2$ which, for $\theta= 1/(4n)$, will decay with $n$. Thus, choosing $\theta \sim n^{-p}$ with $p> 1$ represents an interesting setting for the bound in~\Cref{mlRiskHideAndSeek}, as the plots in~\Cref{fig:sfig1} and~\Cref{fig:sfig2} show.

	\begin{figure*}
		\centering
		\subfloat[$\theta=\left(\frac{1}{n}\right)^{1.5}$]{ %\includegraphics[width=.5\linewidth]{images/p=1.5.pdf}
        \includegraphics[width=.5\linewidth]{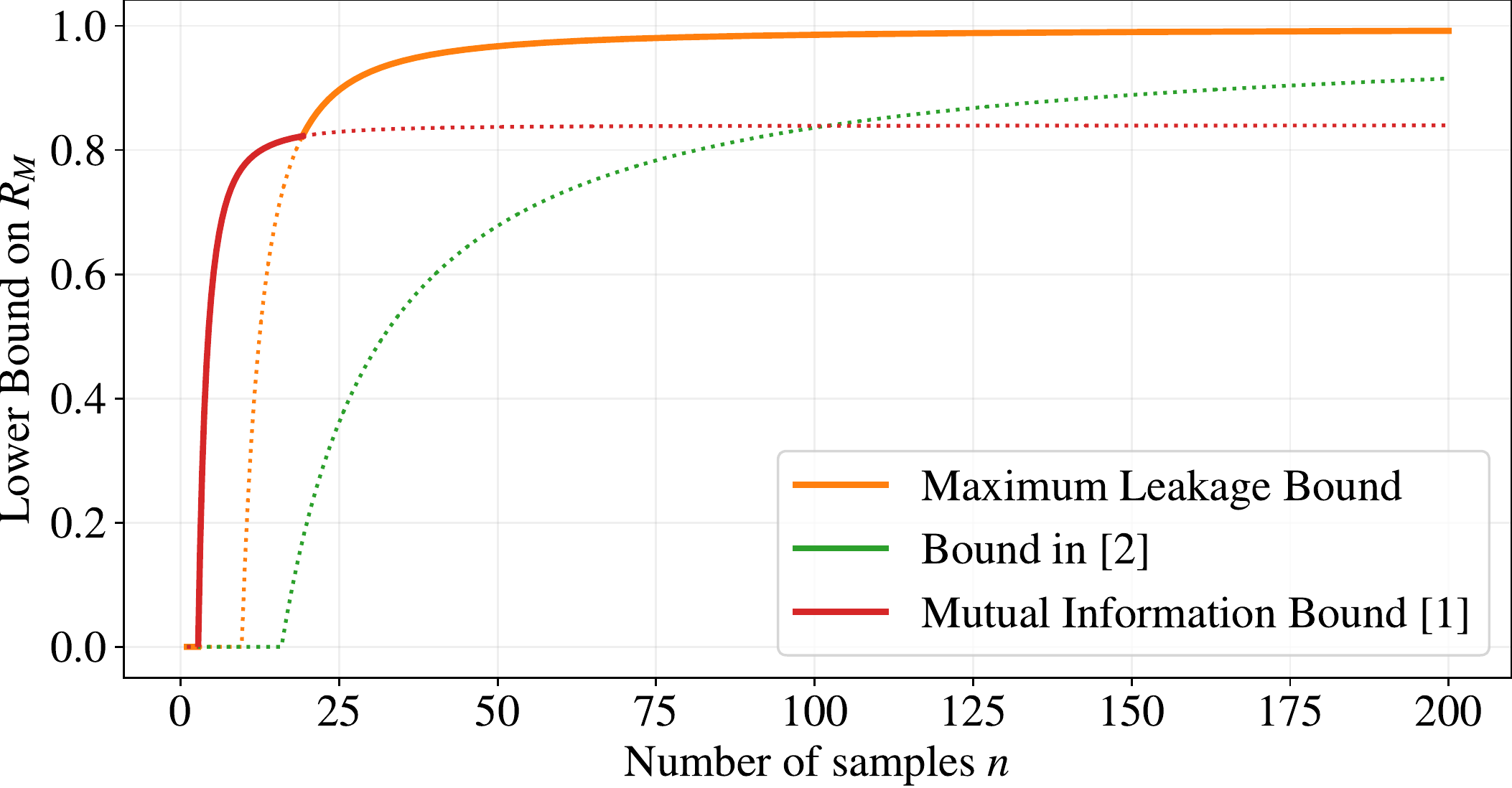}
			\label{fig:sfig1}}
		\subfloat[$\theta = \left(\frac{1}{n}\right)^2$]{ %\includegraphics[width=.5\linewidth]{images/p=2.pdf}
        \includegraphics[width=.5\linewidth]{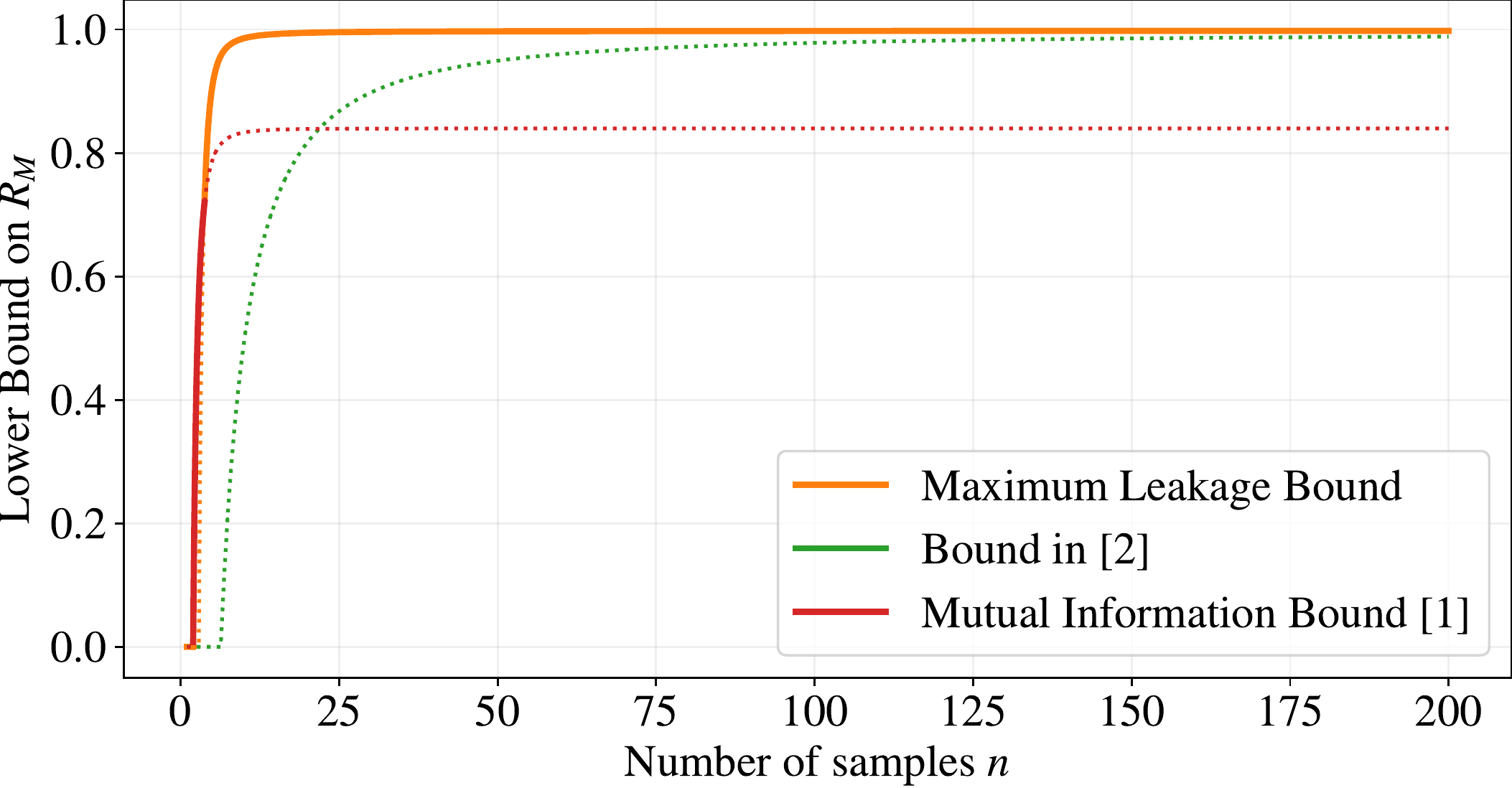}
			\label{fig:sfig2}}
		\\
        \vspace{1em}
		\subfloat[$\theta = 0.01$]{ %\includegraphics[width=.5\linewidth]{images/r=0.01.pdf}
        \includegraphics[width=.5\linewidth]{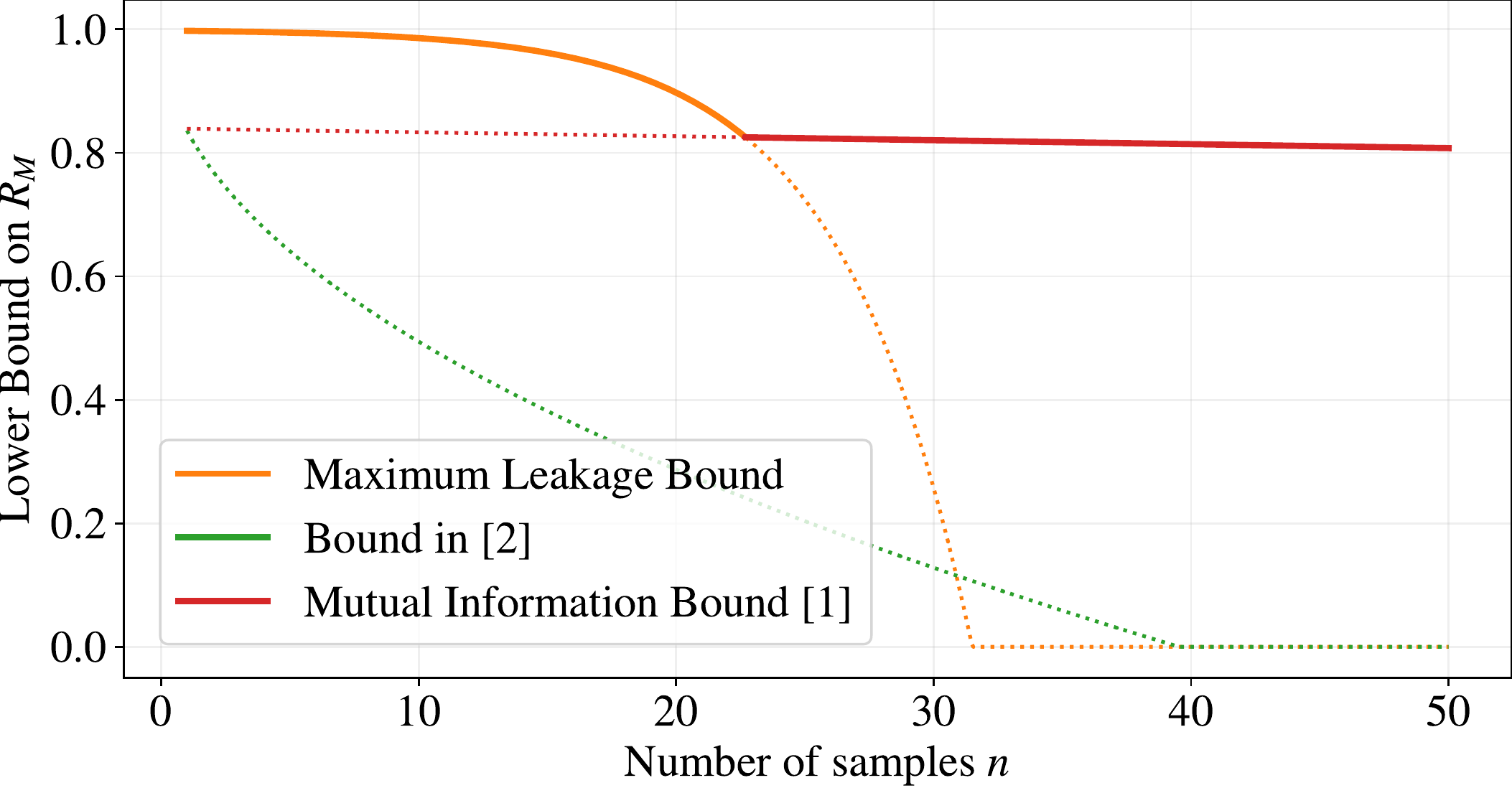}
			\label{fig:sfig3}}
		\subfloat[$\theta = 0.0001$]{ %\includegraphics[width=.5\linewidth]{images/r=0.0001.pdf}
        \includegraphics[width=.5\linewidth]{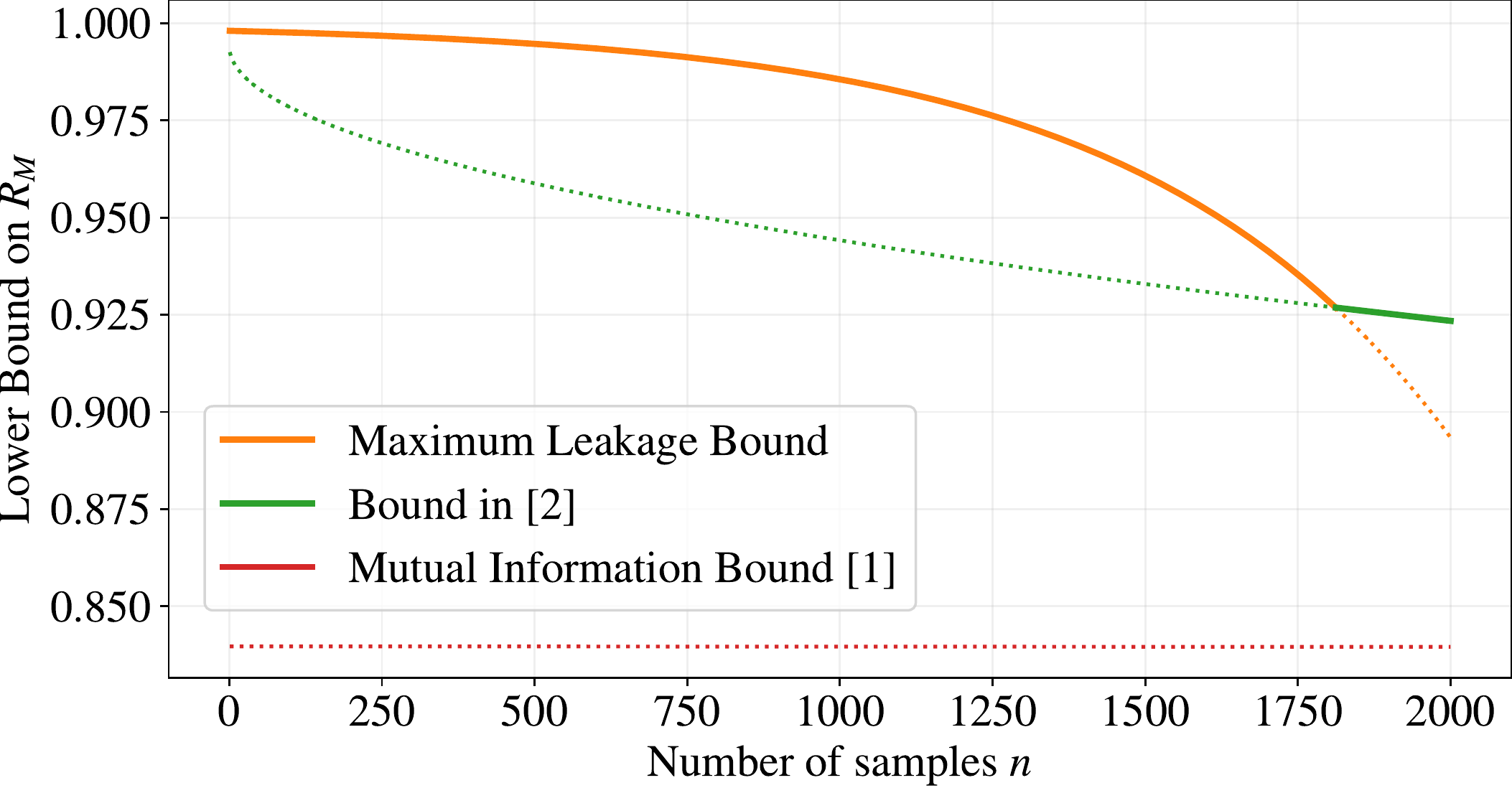}
			\label{fig:sfig4}}
		\caption{Behaviour of~\Cref{mlRiskHideAndSeek} for various values of $\rho$. Solid lines mean that the corresponding lower-bound is the largest.}
		\label{fig:fig}
	\end{figure*}
	
	Thanks to the different behaviour of~\Cref{mlRiskHideAndSeek} (reaching $1$ exponentially fast) we can see a much sharper jump towards $1$ with respect to~\Cref{miRiskHideAndSeek}, which instead plateaus below $1$, and with respect to~\Cref{nipsRiskHideAndSeek} that reaches $1$ more slowly. The growth towards $1$ of~\Cref{mlRiskHideAndSeek} becomes even sharper with larger $p$ and converges towards a specific behaviour at $p \approx 2$. Increasing $p$ any further does not alter the behaviour of the bound meaningfully.
	As for the behaviour of the bound for fixed $\theta$, if $\theta=0.01$. then~\Cref{mlRiskHideAndSeek} provides a larger lower-bound only for $n<25$. If the parameter is brought down to $\theta=0.0001$ then~\Cref{mlRiskHideAndSeek} is larger than~\Cref{miRiskHideAndSeek} for all $n$ but only larger than~\Cref{nipsRiskHideAndSeek} for $n<1850$.
	Regardless of the considerations related to the specific settings, it is interesting how a very simple application of~\Cref{eq:maximalLeakgeResultBayesRisk} can provide a tighter lower-bound. Moreover, in the proof of~\Cref{miRiskHideAndSeek} in~ \citep{bayesRiskRaginsky}, in order to compute $I(W;X)$ an assumption on the distribution of $W$ was necessary and the choice fell on $W$ uniform on $[d]$. With Maximal Leakage, $\ml{W}{X}$ does not depend on the specific distribution over $W$, rendering the bound potentially more general.
	Other divergences could be explored in this setting as well. However, one in general does not have a chain rule for any other $\varphi$-Divergence (or Sibson's $\alpha$-Mutual Information with $\alpha<+\infty$) which is a fundamental step in the proof for Maximal Leakage (cf.~\Cref{eq:hideSeekMLChainRule}). Moreover, some assumption (or maximisation over) $\Pm_W$ would be necessary. In general, some additional machinery would be required in order to employ them in this setting. Given that this is outside of the scope of this work, these approaches will not be explored in this document.

	\appendices
 
\crefalias{section}{appendix}
\crefalias{subsection}{appendix}
	\section{Comparison with similar approaches}\label{app:comparisonWithChen}
		An approach closely connected to the one proposed in here is in~\cite{bayesRiskFInformativity}. The authors therein focused on the notion of $\varphi$-informativity~\cite {fInformativity} and leveraged the Data-Processing inequality of the information measure.
		In particular, $\varphi$-informativities can potentially lead to tighter results than the $\varphi$-Mutual Information considered in this work. Similarly to Sibson's $\alpha$-Mutual Information, they are defined as follows:
		\begin{equation}
			\hat{I}_\varphi(X,Y)= \inf_{\Q_Y} D_\varphi(\Pm_{XY}\|\Pm_X\Q_Y)\leq I_\varphi(X,Y). \label{eq:fInformativity}
		\end{equation}
		Given that the minimum-achieving distribution, $\Q_Y^\star$, is guaranteed to exist in~\Cref{eq:fInformativity} (see~\cite{fInformativity}), one can see that $\hat{I}_\varphi(X,Y)=D_\varphi(\Pm_{XY}\|\Pm_X\Q_Y^\star)$. Consequently, the same steps followed in the proof of~\Cref{thm:fDivBoundBayesRisk} can be undertaken in order to reach a similar result involving $\hat{I}_\varphi$ and $\Pm_X\Q_Y^\star$ rather than $I_\varphi$ and $\Pm_X\Pm_Y$. However, except in some specific settings, the minimum-achieving distribution in~\Cref{eq:fInformativity} does not necessarily admit a closed-form expression~\cite{fInformativity}. As a consequence, the corresponding $\varphi$-Informativity does not admit a closed-form expression. Moreover, another step the authors leveraged in order to achieve ~\cite[Theorem 3.2]{bayesRiskFInformativity}, is the inversion of the resulting binary divergence, leading to a bound which can rarely be expressed in closed form and can only be computed numerically. While a direct comparison between the two approaches would be hard, some similarities are present and hint at the fact that~\cite[Theorem 3.2]{bayesRiskFInformativity} is tighter than~\Cref{thm:fDivBoundBayesRisk}. Indeed, an alternative proof for~\Cref{thm:fDivBoundBayesRisk} also stems from leveraging the DPI of $I_\varphi$ (cf.~\cite[Theorem 3]{fullVersionGeneralization}). However, additional steps are introduced in order to get a closed-form lower bound. Our analysis is designed to retrieve a large family of results which are amenable to analysis and interpretable. This allows us to retrieve lower bounds in closed-form expressions that can be seen to match the upper bounds, up to a constant, in a variety of settings. From a more conceptual standpoint, one could see~\cite[Theorem 3]{fullVersionGeneralization} (and, consequently,~\Cref{thm:fDivBoundBayesRisk}) as a generalisation of H\"older's\footnote{Selecting $\varphi(x)=x^p$ specialises~\Cref{thm:fDivBoundBayesRisk} to~\Cref{thm:lowerBoundHellinger} which can also be proven as an application of H\"older's inequality followed by Markov's inequality, cf.~\cite[Corollary 6]{fullVersionGeneralization}.} inequality to arbitrary convex functionals. This generalisation, which in turn can be seen as a generalisation of Fano's inequality for $\varphi$-Mutual Information, allows us to also encompass divergences from the R\'enyi's family and Sibson's $\alpha$-Mutual Information, which are not $\varphi$-Divergences and are thus excluded from~\cite{bayesRiskFInformativity}. To conclude, let us highlight that our approach, which leverages duality, allows us to provide a single analysis for every type of loss and does not require a separate treatise for $0-1$ losses and more general losses. Consequently, the two approaches for general losses are different and hard to compare.
    \section{Proof of~\Cref{sec:riskLBThroughProbs}}
    \subsection{Proof of~\Cref{thm:sibsMIResultBayesRisk}}\label{app:proofIalpha}
    \begin{proof}
		We have that \begin{align}\Pm_{W\hat{W}}(\ell(W,\hat{W})< \rho) &\leq\left(\sup_{\hat{w}\in\mathcal{\hat{W}}}\Pm_W(\ell(W,\hat{w})< \rho)\right)^{\frac{\alpha-1}{\alpha}}\exp\left(\frac{\alpha-1}{\alpha}I_\alpha(W,\hat{W})\right) \label{HolderBound}
			\\
			&=\exp\left(\frac{\alpha-1}{\alpha}\left(I_\alpha(W,\hat{W}) + \log(L_W(\rho))\right) \right)\\ &\leq \exp\left(\frac{\alpha-1}{\alpha}\left(I_\alpha(W,X) + \log(L_W(\rho))\right) \right). \label{DPI}\end{align}
		\Cref{HolderBound} follows from~\Cref{sibsMIBoundCor},~\Cref{DPI} follows from the Data-Processing Inequality for $I_\alpha$ and the Markov Chain $W-X-Y-\hat{W}$. Moreover, %, starting from~\Cref{markov} and 
		using Markov's inequality one has that  \begin{align}R_B&\geq \rho \cdot \Pm_{W\hat{W}}(\ell(W,\hat{W})\geq \rho) \\ 
			&= \rho \cdot (1-\Pm_{W\hat{W}}(\ell(W,\hat{W})< \rho)). \label{markovLowerBound}\end{align} The statement follows from lower bounding~\Cref{markovLowerBound} using~\Cref{DPI}.
	\end{proof}
    \subsection{Proof of~\Cref{thm:fDivBoundBayesRisk}}\label{app:proofDiv}
    \begin{proof}
    From the variational representation for $\varphi$-Divergences (see~\Cref{eq:varReprfDiv}), given $\Pm_{W\hat{W}}$, for every function $f$ in the respective space (defined in~\Cref{thm:varReprfDiv}) one has that:
    \begin{equation}
        I_\varphi(W,\hat{W}) = D_\varphi(\Pm_{W\hat{W}}\|\Pm_W\Pm_{\hat{W}}) \geq \Pm_{W\hat{W}}(f) - \Pm_W\Pm_{\hat{W}}(\varphi^\star(f)).\label{eq:varReprJoint}
    \end{equation}
    ~\Cref{eq:varReprJoint} allows us to relate the expected value of any function $f:\mathcal{W}\times \mathcal{\hat{W}}$ under the joint with $I_\varphi(W,\hat{W})$ and the corresponding Legendre-Fenchel dual. Our purpose is to provide a lower-bound on the expected loss $\ell$ hence we will select $f=\tilde{\lambda}(\rho-\ell)$ with $\rho,\tilde{\lambda}>0$. Moreover, given the non-negativity of $\ell$ one can also see that $\ell \geq \rho \mathbbm{1}_{\{\ell \geq \rho\}}$ (\textit{i.e.}, Markov's Inequality in its functional form). Hence, plugging our choice of $f$ in~\Cref{eq:varReprJoint} the following chain of inequality follows:
    \begin{align}
       \tilde{\lambda} \Pm_{W\hat{W}}(\ell) &\geq \tilde{\lambda}\rho -I_\varphi(W,\hat{W}) - \Pm_W\Pm_{\hat{W}}(\varphi^\star(\tilde{\lambda}(\rho-\ell)) \\
       &\geq \tilde{\lambda}\rho -I_\varphi(W,\hat{W}) - \Pm_W\Pm_{\hat{W}}(\varphi^\star(\tilde{\lambda}\rho(1-\mathbbm{1}_{\{\ell\geq \rho\}}))) \label{eq:nonDecreasability} \\&=\tilde{\lambda}\rho -I_\varphi(W,\hat{W}) - \Pm_W\Pm_{\hat{W}}(\varphi^\star(\tilde{\lambda}\rho\mathbbm{1}_{\{\ell< \rho\}})) \\ &=  \tilde{\lambda}\rho -I_\varphi(W,\hat{W}) - \Pm_W\Pm_{\hat{W}}(\{\ell < \rho\})\cdot\varphi^\star(\tilde{\lambda}\rho) - \Pm_W\Pm_{\hat{W}}(\{\ell \geq \rho\})\cdot\varphi^\star(0),
    \end{align}
    where~\Cref{eq:nonDecreasability} follows by the monotonicity of $\varphi^\star$ which can be seen as stemming from the strict convexity and the monotonicity of $\varphi$. Indeed, if $\varphi$ is strictly convex then ${\varphi^\star}'(t) = {\varphi'}^{-1}(t)$ for every $t \in \text{Im}(\varphi')$~\cite[Theorem 26.5]{rockafellar-1970a}. Since $\varphi$ is monotone non-decreasing on the positive axis, one has that $\varphi'(t)\geq 0$ on $[0,+\infty]$. Accordingly, the inverse of $\varphi'$ will also be non-negative on $[0,+\infty]$, which implies the non-negativity of ${\varphi^\star}'$ and, therefore, the monotonicity of $\varphi^\star$. A similar argument shows the monotonicity of $\varphi^\star$ when $\varphi$ is non-increasing.
    Then, dividing both sides by $\tilde{\lambda}$ and selecting $\tilde{\lambda}= \frac1\rho \lambda$ with $\lambda>0$ one recovers the following:
    \begin{align}
       \Pm_{W\hat{W}}(\ell) &\geq \sup_{\lambda>0} \rho\left( 1- \frac{ I_\varphi(W,\hat{W})  + \Pm_W\Pm_{\hat{W}}(\{\ell \geq \rho\})\cdot\varphi^\star(0) + \Pm_W\Pm_{\hat{W}}(\{\ell < \rho\})\cdot\varphi^\star(\lambda)}{\lambda} \right) \\&= \rho\left( 1- \inf_{\lambda>0} \frac{ I_\varphi(W,\hat{W}) + \Pm_W\Pm_{\hat{W}}(\{\ell \geq \rho\})\cdot\varphi^\star(0) + \Pm_W\Pm_{\hat{W}}(\{\ell < \rho\})\cdot\varphi^\star(\lambda) }{\lambda} \right) \\
       &= \rho\left( 1-\Pm_W\Pm_{\hat{W}}(\{\ell < \rho\}) \inf_{\lambda>0} \frac{ \frac{I_\varphi(W,\hat{W}) + \Pm_W\Pm_{\hat{W}}(\{\ell \geq \rho\})\cdot\varphi^\star(0)}{\Pm_W\Pm_{\hat{W}}(\{\ell < \rho\})} + \varphi^\star(\lambda) }{\lambda} \right) \\
       &= \rho\left(1-\Pm_W\Pm_{\hat{W}}(\{\ell < \rho\}) \varphi^{-1}\left(\frac{I_\varphi(W,\hat{W}) + \Pm_W\Pm_{\hat{W}}(\{\ell \geq \rho\})\cdot\varphi^\star(0)}{\Pm_W\Pm_{\hat{W}}(\{\ell < \rho\})}\right)\right) \label{eq:inverseVarphi},
    \end{align}
    where~\Cref{eq:inverseVarphi} follows from the same argument as in~\cite[Theorem 13]{thesis}.\\ \Cref{eq:nonDecreasability} is the step of the proof which is relevant with respect to~\Cref{rmk:onTheDual}. In particular, the choice of $f$, along with the non-decreasability of $\varphi$ allowed us to leverage the functional form of Markov's inequality and, consequently, to upper-bound the dual of $D_\varphi(\cdot\| \Pm_W\Pm_{\hat{W}})$. Upper-bounding the dual is crucial in order to achieve a bound of the form of~\Cref{eq:inverseVarphi}. In order to prove the result for $\varphi$ non-increasing one has to select $f=-
    \tilde{\lambda}\ell$ leverage Markov's inequality and select $\tilde{\lambda}=-\frac1\rho\lambda$ with $\lambda<0$. The result then follows from the same argument as in~\cite[Theorem 13]{thesis} \textit{i.e.}, from selecting $\lambda = {\varphi'}(\varphi^{-1}(c))$ with $c=\frac{I_\varphi(W,\hat{W})+\varphi^\star(0)\Pm_W\Pm_{\hat{W}}(E)}{\Pm_W\Pm_{\hat{W}}(E^c)}$ and $E=\{\ell< \rho\}$.
	\end{proof}
    \subsection{Proof of~\Cref{thm:lowerBoundHellinger}}\label{app:proofHellinger}
    \begin{proof}
		The statement follows from~\Cref{thm:fDivBoundBayesRisk} with $\varphi(x)= \frac{x^p-1}{p-1}$. %(or, from~\Cref{generalHellingerBound}). 
		Hence, for every estimator $\hat{W}=\phi(X^n)$,
		\begin{align}
			\Pm_{W\hat{W}}(\ell(W,\hat{W})\geq \rho) &\leq L_W(\hat{W},\rho)\cdot \varphi^{-1}\left(\frac{I_\varphi(W, \hat{W})+(1-L_W(\hat{W},\rho))\cdot\varphi^\star(0)}{L_W(\hat{W},\rho)}\right)  \\&\leq L_W(\rho)^{\frac{p-1}{p}}\left((p-1)\mathcal{H}_p(W, X)+1\right)^{\frac1p}, \label{eq:lastStepProofHellingerBayesRisk}
		\end{align}
		where~\Cref{eq:lastStepProofHellingerBayesRisk} follows from the data-processing inequality for $\varphi$-divergences. One thus retrieves that for every estimator $\hat{W}$
		\begin{equation}
			\Pm_{W\hat{W}}(\ell(W, \hat{W})) \geq\rho\left(1-L_W(\rho)^{\frac{p-1}{p}}\left((p-1)\mathcal{H}_p(W, X)+1\right)^{\frac1p}\right). \label{lowerBoundTemp}
		\end{equation}
		Since the right-hand side of~\Cref{lowerBoundTemp} is independent of $\hat{W}=\phi(X)$ one can use it to lower-bound the risk $R$. 
	\end{proof}
 \subsection{Proof of~\Cref{thm:lowerBoundGeneralHockeyStick}}\label{app:proofHockeyStick}
    \begin{proof}
		Let $\varphi(x) = \max\{0, \zeta x-\gamma\}-\max\{0, \zeta-\gamma\}$ in~\Cref{thm:fDivBoundBayesRisk}, along with the fact that $\varphi^{-1}(y) = \frac{y+\max\{0,\zeta-\gamma\}+\gamma}{\zeta}$ for $y>0$ and $\varphi^\star(0) = \max\{0, \zeta-\gamma\}$ one has that for every estimator $\hat{W}=\phi(X^n)$,
		\begin{align}
			\Pm_{W\hat{W}}(\ell(W, \hat{W})) &\geq  
			\rho\left(1-\frac{E_{\gamma,\zeta}(W, \hat{W}) + \gamma L_W(\hat{W}, \rho)+ \max\{0,\zeta-\gamma\}}{\zeta}\right) \\ 
			&\geq \rho\left(1-\frac{E_{\gamma,\zeta}(W, X) + \gamma L_W(\rho)+ \max\{0,\zeta-\gamma\}}{\zeta}\right). \label{eq:eGammaDelta}
		\end{align}
		Since~\Cref{eq:eGammaDelta} is independent of $\hat{W}=\phi(X)$ one can use it to lower-bound the risk $R$.
	\end{proof}
        \section{Computations for~\Cref{sec:bayesRiskExample}} 
        \subsection{Maximisation over $\rho$}\label{app:rho_maximization}
        The bounds considered in the first three examples have the following form
\begin{equation} \label{equation:rho_equation}
\sup_{\rho>0} \rho(1-c\rho^t-b),
\end{equation}
for some $c, t, b\geq 0$. Letting $ h(\rho):=\rho(1-c\rho^t-b)$, the optimal value $\rho_\star$ is found by setting $h^\prime(\rho_\star) = 0$, which yields
\begin{equation}
    1-(t+1)c\rho_\star^t-b=0 \iff \rho_\star = \left(\frac{1-b}{(t+1)c}\right)^{\frac{1}{t}}. 
\end{equation}
Since $h^{\prime\prime}(\rho_\star) = -t(t+1)c\rho_\star^{t-1} \leq 0$, this ensures $\rho_\star$ is a maximum. Substituting $\rho^\star$ back in \Cref{equation:rho_equation}, we can express the lower bound as
\begin{equation}
    \sup_{\rho>0} \rho(1-c\rho^t-b)=\frac{t}{c^{\frac1t}}\left(\frac{1-b}{t+1}\right)^{1+\frac1t}.\label{rhoStar}
\end{equation}
        \subsection{~\Cref{sec:ex1}}
        \subsubsection{Maximal Leakage}\label{app:mlBernoulliBias}
        In this setting one has that $$\Pm_{X^n|W=w}(x^n) = w^k(1-w)^{(n-k)}$$ where $k=\sum_{i=1}^n$ is the hamming weight of $x^n$. As per assumption, $\Pm_W(w)=1$ if $0\leq w\leq 1$ and, consequently, one has that $$\Pm_{W|X^n=x^n}(w)= (n+1)\binom{n}{k}(1-w)^{n-k}w^k.$$
	One can thus compute Maximal Leakage in this setting:
	\begin{align}
		\ml{W}{X^n} &= \log \sum_{x^n} \max_w \Pm_{X^n|W=w}(x^n)\\
		&= \log \sum_{k=0}^n \binom{n}{k} \max_w w^k(1-w)^{n-k} \\
		&= \log \sum_{k=0}^n \binom{n}{k} \left(\frac{k}{n}\right)^k\left(1-\frac{k}{n}\right)^{n-k} \\
		&\leq \log \left( 2+ \sum_{k=1}^{n-1} \sqrt{\frac{n}{2\pi k(n-k)}} \right) \label{stirling}\\
		&\leq \log \left( 2+ \sqrt{\frac{\pi n}{2}}\right). 
	\end{align}
    \subsubsection{Sibson's $\alpha$-Mutual Information}\label{app:sibsonBernoulliBias}
    For Sibson's $\alpha$-Mutual Information with $\alpha>1$, one has that:
    \begin{align}
		\exp\left(\frac{\alpha-1}{\alpha}I_\alpha(W,X^n))\right) &= \E\left[\E^{\frac1\alpha}\left[\left(\frac{\Pm_{X^n|W}}{\Pm_{X^n}}\right)^{\alpha}\bigg| X^n\right]\right] \\
        &= \sum_{x^n} \Pm_{X^n}(x^n)\left(\int_0^1 \Pm_W(w)\left(\frac{\Pm_{W|X^n=x^n}(w)}{\Pm_W(w)}\right)^{\alpha}dw\right)^{\frac1\alpha}\\
        &= \sum_{x^n} \Pm_{X^n}(x^n)\left(\int_0^1 \left(\Pm_{X^n|W=w}(x^n)\right)^{\alpha}dw\right)^{\frac1\alpha}\\
        &= \sum_{k=0}^n \binom{n}{k}\frac{1}{(n+1)\binom{n}{k}}\left(\int_0^1 \left((n+1)\binom{n}{k}w^k(1-w)^{n-k}\right)^{\alpha}dw\right)^{\frac1\alpha}\\
        &= \sum_{k=0}^n \binom{n}{k}\left(\int_0^1 \left(w^k(1-w)^{n-k}\right)^{\alpha}dw\right)^{\frac1\alpha}\\
        &= \sum_{k=0}^n \binom{n}{k}\left(\frac{\Gamma(k\alpha + 1)\Gamma((n-k)\alpha + 1)}{\Gamma(n\alpha+2)}\right)^{\frac1\alpha},\label{eq:sibsonLastStep}
	\end{align}
    where~\Cref{eq:sibsonLastStep} uses the identity relating the Beta function with the Gamma function \textit{i.e.},
	\begin{equation}
		\mathrm{Beta}(x, y) = \int_0^1 w^{x-1}(1-w)^{y-1}dw = \frac{\Gamma(x)\Gamma(y)}{\Gamma(x+y)},
	\end{equation}
    so that
    \begin{equation}
		\int_0^1 w^{k\alpha}(1-w)^{(n-k)\alpha}dw = \frac{\Gamma(k\alpha + 1)\Gamma((n-k)\alpha + 1)}{\Gamma(n\alpha+2)}.\label{eq:beta_gamma_relation}
	\end{equation}
        \subsubsection{Hellinger $p$} \label{app:hellBernoulliBias}For the Hellinger$-p$ divergence with $p>1$, one has that:
	\begin{align}
		((p-1)\mathcal{H}_p(W,X^n)+1) &= 
		\left\lVert\frac{d\Pm_{WX^n}}{d\Pm_W\Pm_{X^n}} \right\rVert_{L_p(\Pm_W\Pm_{X^n})}^p\\
		&=  \sum_{x^n} \Pm_{X^n}(x^n)\int_0^1 \Pm_W(w)\left(\frac{\Pm_{W|X^n=x^n}(w)}{\Pm_W(w)}\right)^pdw\\
        &=\sum_{x^n} \Pm_{X^n}(x^n)\int_0^1 \left(\Pm_{X^n|W=w}(x^n)\right)^{p}dw\\
        &= \sum_{k=0}^n \binom{n}{k}\frac{1}{(n+1)\binom{n}{k}}\int_0^1 \left((n+1)\binom{n}{k}w^k(1-w)^{n-k}\right)^{p}dw\\
        &= (n+1)^{p-1}\sum_{k=0}^n \binom{n}{k}^p \int_0^1 \left(w^k(1-w)^{(n-k)}\right)^pdw \\
        &=(n+1)^{p-1}\sum_{k=0}^n \binom{n}{k}^p\frac{\Gamma(kp+1)\Gamma((n-k)p+1)}{\Gamma(np+2)},\label{eq:hellingerLastStep}
	\end{align}
	where~\Cref{eq:hellingerLastStep} follows from~\Cref{eq:beta_gamma_relation}.
    For the special case $p=2$, we get
    \begin{align}
    \chi^2(W,X^n)+1 &= (n+1)\sum_{k=0}^n \binom{n}{k}^2\frac{(2k)!(2(n-k))!}{(2n+1)!}\\&=\frac{n+1}{(2n+1)} \sum_{k=0}^n \frac{(n!)^2(2k)!(2(n-k))!}{(k!)^2((n-k)!)^2(2n)!}\\&= \frac{n+1}{(2n+1)\binom{2n}{n}} \sum_{k=0}^n \binom{2k}{k}\binom{2(n-k)}{n-k} \\&= \frac{n+1}{2n+1}\cdot \frac{4^n}{\binom{2n}{n}},\label{equation:chi_square_before_ub}
\end{align}
where in~\Cref{equation:chi_square_before_ub} we use the result in~\cite[Eq. (5.39), p.187]{concrete_mathematics} stating that $\sum_{k=0}^n \binom{2k}{k}\binom{2(n-k)}{n-k} = 4^n$.
    \subsection{~\Cref{sec:gaussianExample}}
    \subsubsection{Hellinger $p$} \label{app:hellGaussian}
    For the Hellinger$-p$ divergence with $p>1$, one has that:
    \begin{align}
		((p-1)\mathcal{H}_p(W,X^n)+1) &= 
		\left\lVert\frac{d\Pm_{WX^n}}{d\Pm_W\Pm_{X^n}} \right\rVert_{L_p(\Pm_W\Pm_{X^n})}^p\\
		&= \int_{\mathbb{R}} \int_{\mathbb{R}} \Pm_{W}(w)\Pm_{X}(x) \left(\frac{\Pm_{X|W=w}(x)}{\Pm_{X}(x)}\right)^p dwdx \\
    &=\int_{\mathbb{R}} \mathcal{P}_{X}(x)^{1-p} \int_{\mathbb{R}} \mathcal{P}_{W}(w) \mathcal{P}_{X|W=w}(x)^pdwdx.\label{eq:hellinger_with_Gp}
	\end{align}

    Focusing on the innermost integral (which we denote as $I_p(x)$), one has
    \begin{align}
    I_p(x) &\vcentcolon= \int_{\mathbb{R}}\mathcal{P}_{{W}}({w})\mathcal{P}_{{X}|{W}={w}}({x})^{p} d{w} \\
    &= \left(\frac{(2\pi\sigma^2)^{-p}}{2\pi \sigma_W^2}\right)^{\frac{1}{2}} \int_{\mathbb{R}} e^{-\frac{w^2}{2\sigma_W^2} - \frac{p (w-x)^2}{2\sigma^2}}d{w}\\
    &=\left(\frac{(2\pi\sigma^2)^{-p}}{2\pi \sigma_W^2}\right)^{\frac{1}{2}}\int_{\mathbb{R}} e^{-\frac{1}{2\sigma^2} \left(px^2-2p{x}{w} + \left(\frac{\sigma^2}{\sigma_W^2}+p\right)w^2\right)}d{w}\\
    &= \left(\frac{(2\pi\sigma^2)^{-p}}{2\pi \sigma_W^2}\right)^{\frac{1}{2}}e^{\frac{-p\cdot x^2}{2\sigma^2}}\int_{\mathbb{R}} e^{-\frac{1}{2\sigma^2} \left(-2p{x}{w} + \left(\frac{\sigma^2}{\sigma_W^2}+p\right)w^2\right)}d{w}\label{equation:I_p}.
\end{align}
Adding and subtracting %The next step is to add and remove, in the exponent inside the integral, 
$cx^2$ with $c=-p\left(1+p\frac{\sigma_W^2}{\sigma^2}\right)^{-1}$ in the exponent inside the integral in \Cref{equation:I_p} leads to
\begin{align}
    I_p(x) &= \left(\frac{(2\pi\sigma^2)^{-p}}{2\pi \sigma_W^2}\right)^{\frac{1}{2}}e^{{\frac{c{x}^2}{2\sigma^2}}}\int_{\mathbb{R}} e^{-\frac{\frac{\sigma^2}{\sigma_W^2}+p}{2\sigma^2} \left({w}-\sqrt{\frac{p+c}{\frac{\sigma^2}{\sigma_W^2}+p}}{x}\right)^2}d{w} \\
    &= \left(\frac{(2\pi\sigma^2)^{-p}}{2\pi \sigma_W^2}\right)^{\frac{1}{2}}\exp\left(-{\frac{p{x}^2}{2\sigma^2\left(1+p\frac{\sigma_W^2}{\sigma^2}\right)}}\right) \left(2\pi \frac{\sigma^2}{\frac{\sigma^2}{\sigma_W^2}+p}\right)^{\frac{1}{2}}\\
    &=\left(2\pi\sigma^2\right)^{-\frac{p}{2}} \left(1+p\frac{\sigma_W^2}{\sigma^2}\right)^{-\frac{1}{2}} \exp\left(-{\frac{p{x}^2}{2\sigma^2\left(1+p\frac{\sigma_W^2}{\sigma^2}\right)}}\right).
\end{align}
Finally, plugging the value of $I_p$ back in  \eqref{eq:hellinger_with_Gp}, we retrieve that:
\begin{align}
    &(p-1)\mathcal{H}_p(W, X)+1 \notag \\
    &= \int_{\mathbb{R}} \mathcal{P}_{{X}}({x})^{1-p} \frac{1}{(2\pi \sigma^2)^{\frac{p}{2}}} e^{-\frac{p{x}^2}{2\left(\sigma^2+p\sigma_W^2\right)}} \left(1+p\frac{\sigma_W^2}{\sigma^2}\right)^{-\frac{1}{2}}d{x}\\
    &= \frac{\left(1+\frac{\sigma_W^2}{\sigma^2}\right)^{\frac{d(p-1)}{2}}}{(2\pi \sigma^2)^{\frac{1}{2}}\left(1+p\frac{\sigma_W^2}{\sigma^2}\right)^{\frac{1}{2}}}\int_{\mathbb{R}} e^{\frac{(p-1){x}^2}{2\left(\sigma^2+\sigma_W^2\right)}-\frac{p{x}^2}{2\left(\sigma^2+p\sigma_W^2\right)}} d{x} \\
    &= \frac{\left(1+\frac{\sigma_W^2}{\sigma^2}\right)^{\frac{(p-1)}{2}}}{(2\pi \sigma^2)^{\frac{1}{2}}\left(1+p\frac{\sigma_W^2}{\sigma^2}\right)^{\frac{1}{2}}}\int_{\mathbb{R}} e^{-\frac{{x}^2}{2}\left(\frac{1-p}{\sigma^2+\sigma_W^2}+\frac{p}{\sigma^2+p\sigma_W^2}\right)} d{x} \\
    &= \frac{\left(1+\frac{\sigma_W^2}{\sigma^2}\right)^{\frac{(p-1)}{2}}}{(2\pi \sigma^2)^{\frac{1}{2}}\left(1+p\frac{\sigma_W^2}{\sigma^2}\right)^{\frac{1}{2}}} \left(\frac{2\pi}{\frac{1-p}{\sigma^2+\sigma_W^2}+\frac{p}{\sigma^2+p\sigma_W^2}}\right)^{\frac{1}{2}} \\
    &= \frac{\left(1+\frac{\sigma_W^2}{\sigma^2}\right)^{\frac{(p-1)}{2}}}{\left(\sigma^2+p\sigma_W^2\right)^{\frac{1}{2}}} \left(\frac{1}{\frac{1-p}{\sigma^2+\sigma_W^2}+\frac{p}{\sigma^2+p\sigma_W^2}}\right)^{\frac{1}{2}} \\
    &= \left(\frac{\left(1+\frac{\sigma_W^2}{\sigma^2}\right)^{p-1}}{\frac{(1-p)\left(\sigma^2+p\sigma_W^2\right)}{\sigma^2+\sigma_W^2}+p} \right)^{\frac{1}{2}} \\
    &=  \left(\frac{\left(1+\frac{\sigma_W^2}{\sigma^2}\right)^p}{1 + (2-p)p\frac{\sigma_W^2}{\sigma^2}} \right)^{\frac{1}{2}}.
\end{align}
    
        \section{Other approaches}\label{sec:otherApproaches}
        \subsection{Inverting the roles}\label{sec:invertingBayesRisk}
	The Sibson's $\alpha$-mutual information is an asymmetric quantity. A natural question is: can one provide a result similar to~\Cref{thm:sibsMIResultBayesRisk} involving $I_\alpha(X,W)$ instead?
	Indeed, by inverting the roles of $W$ and $\hat{W},$ such a bound can be given but it will involve the small ball probability for $\hat{W}$ \textit{i.e.},
	\begin{align}
		L_{\hat{W}}(\rho) & = \sup_{w} \Pm_{\hat{W}}(\ell(w,\hat{W})\geq \rho).
	\end{align}
	This quantity hinges on the marginal distribution of $\hat{W},$ which, in turn, depends on the estimator used. 
	In terms of $L_{\hat{W}}(\rho),$ one can give the following general bound:
	\begin{lemma}\label{invertedResults}
		Consider the Bayesian framework described in~\Cref{sec:bayesianFramework}. The following holds for every $\alpha>1$ and $\rho>0$:
		\begin{equation}
			R_B\geq \rho\left(1- \exp\left(\frac{\alpha-1}{\alpha}\left(I_\alpha(X,W) + \log(L_{\hat{W}}(\rho))\right) \right)\right). \label{sibsMILowerBoundInverted}
		\end{equation}
		Moreover, taking the limit of $\alpha\to\infty$ one has:
		\begin{equation}
			R_B\geq \rho\left(1- \exp\left(\ml{X}{W} + \log(L_{\hat{W}}(\rho))\right)
			\right). \label{MLLowerBoundInverted}
		\end{equation}
	\end{lemma}
	To apply this lemma in concrete cases, one needs to compute or upper bound the small ball probability $L_{\hat{W}}(\rho).$
	Leveraging basic properties of the estimator, one can sometimes bound it. For example, if the estimator is a linear function of the noisy observations one can leverage results related to L\'evy's concentration
	functions of sums of independent random variables. \textit{E.g.}, if $Y_1,\ldots,Y_m$ are uncorrelated and have log-concave distributions, then for every $\rho\geq 0$~\cite[Theorem 1.1]{levyLittleSums},
	\begin{equation} L_{\sum_i^m Y_i}(\rho) \leq \frac{2\rho}{\sqrt{\text{Var}(\sum_{i=1}^m Y_i) + \rho^2/3}} = \frac{2\rho}{\sqrt{m \text{Var}( Y_1) + \rho^2/3}} \label{smallBallLinear}.\end{equation}  More general statements can be made, assuming $\phi(Y^m) = \sum_{i=1}^m a_iY_i$ under different constraints over $a_i$~\citep{smallBallai}.
	To appreciate the promise of this approach, let us also discuss the behaviors of $I_\alpha(W,X)$ and $I_\alpha(X,W)$. More specifically, let us consider again the ``Hide-and-Seek'' problem. Assuming, as in~\cite[Example 12]{bayesRiskRaginsky}, that $\Pm_W$ is uniform over $[d]$, one has that
	\begin{equation}
		\ml{X^{n\times m}}{W}=\log\frac{d(1/2+\rho)}{(d-1)(1/2-\rho)+(1/2+\rho)}=\log\kappa(d,\rho)< \log d.
	\end{equation} In case $\rho$ and $d$ are constant and the estimator $\phi$ is a linear combination of the observations, using~\Cref{smallBallLinear} in~\Cref{invertedResults} one gets: \begin{equation}R_B\geq \rho\left(1-\frac{\kappa(d,\rho)2\rho}{\sqrt{m \text{Var}( Y_1) + \rho^2/3}}\right).\end{equation} This lower bound approaches $\rho$ as $m$ grows, rather than providing the trivial lower bound of $0$, as it happens in~\Cref{mlRiskHideAndSeek}.\\
	The assumptions required, along with the need of specifying a prior over $W$, clearly restrict the domain of applicability of~\Cref{invertedResults} with respect to~\Cref{thm:sibsMIResultBayesRisk} and~\Cref{eq:maximalLeakgeResultBayesRisk}. However, this approach can provide results in settings where~\Cref{thm:sibsMIResultBayesRisk} and~\Cref{eq:maximalLeakgeResultBayesRisk} become vacuous.
	
	\subsection{Lower-Bounding the Risk Directly}\label{sec:directLowerBoundRisk}
	An alternative route can be undertaken that does not use Markov's inequality as a first step and can possibly lead to tighter bounds. Since our purpose is to provide \textit{lower bounds} on the Risk (essentially, an inner-product between the joint measure of the parameter and the estimation and the loss function, $\langle \mathcal{P}_{W\hat{W}},\ell\rangle$) one can also consider the application of reverse H\"older's inequality in order to directly lower-bound the Risk. %(without requiring an application of Markov's inequality).
Consider $\alpha<1$, the following result can be easily proven: \begin{corollary}\label{thm:sibsMIDirectResultBayesRisk}
		Consider the Bayesian framework described in~\Cref{sec:bayesianFramework}. The following must hold for every $\alpha,\alpha'<1$
  \begin{align}
    \Pm_{W,\hat{W}}(\ell) \geq & \Pm_{\hat{W}}^{\frac{1}{\beta'}}\left(\Pm_W^{\frac{\beta'}{\beta}}\left(\ell^\beta\right)\right)\cdot\Pm_{\hat{W}}^{\frac{1}{\alpha'}}\left(\Pm_W^{\frac{\alpha'}{\alpha}}\left(\left(\frac{d\Pm_{W\hat{W}}}{d\Pm_W\Pm_{\hat{W}}}\right)^{\alpha}\right)\right),~\label{eq:normDirectResultBayesRisk}
	\end{align}
	where $\frac{1}{\alpha}+\frac{1}{\beta} = 1 = \frac{1}{\alpha'}+\frac{1}{\beta'}$ and $\alpha,\alpha'<1$. 
 Moreover, if one takes the limit of $\alpha' \to 1^-$, which implies $\beta' \to-\infty$, then one recovers the following with $0<\alpha<1$:
		\begin{equation}
			R_B\geq  \essinf_{\Pm_{\hat{W}}} \left(\Pm_W^{\frac{1}{\beta}}\left(\ell(W,\hat{W})^\beta\right)\right)\cdot\exp\left(\frac{\alpha-1}{\alpha}I_\alpha(W,X)\right).\label{eq:sibsMIDirectResultBayesRisk}
		\end{equation}
	\end{corollary}
	\begin{proof}
	The proof of~\Cref{eq:normDirectResultBayesRisk} follows from the same proof of~\Cref{alphaExpBound} (cf.~\cite[Theorem 15]{thesis}) with $f=\ell$ but using reverse H\"older's inequality rather than regular H\"older's inequality.  Considering the limit of $\alpha' \to 1^-$ in~\Cref{eq:normDirectResultBayesRisk} one recovers the following:
		\begin{equation}
			\Pm_{W\hat{W}}(\ell(W,\hat{W})) \geq  \essinf_{\Pm_{\hat{W}}} \left(\Pm_W^{\frac{1}{\beta}}\left(\ell(W,\hat{W})^\beta\right)\right)\cdot\exp\left(\sign(\alpha)\cdot\frac{\alpha-1}{\alpha}I_\alpha(W,\hat{W})\right) \label{eq:lowerBoundExpValueSibsIalpha}
		\end{equation}
		Now, if $0<\alpha<1$ then $\frac1\beta<0$. By the Data-Processing Inequality for $I_\alpha$ with $0<\alpha<1$ (along with the negativity of $\frac1\beta$) one has that
		\begin{align}
			\exp\left(\sign(\alpha)\cdot\frac{\alpha-1}{\alpha}I_\alpha(W,\hat{W})\right) &= \exp\left(\frac{1}{\beta}I_\alpha(W,\hat{W})\right) \geq \exp\left(\frac{1}{\beta}I_\alpha(W,X)\right). \label{eq:dpiIalphaPosSmallerOne}
		\end{align}
	
		The lower-bound on the Risk follows by noticing that the right-hand side~\Cref{eq:lowerBoundExpValueSibsIalpha} can be rendered independent of $\hat{W}$ for every $\alpha<1$ (\textit{i.e.}, it will only depend on the support of $\hat{W}$ through the $\essinf$) via~\Cref{eq:dpiIalphaPosSmallerOne}.
	\end{proof}
	\begin{remark}[Extending to $\alpha<0$]
	    One could also extend the result to $\alpha<0$ (which implies $0<\beta<1$), however, this would lead to a notion of $I_\alpha$ for $\alpha<0$ (cf.~\cite{sibsonIalphaNegative}) which is outside the scope of this work. However, in that case one would have the following interesting limiting behavior when $\alpha\to -\infty$: 
      \begin{align}
          \Pm_{W\hat{W}}(\ell(W,\hat{W})) &\geq \left(\essinf_{\Pm_{\hat{W}}} \Pm_W\left(\ell(W,\hat{w})\right)\right)\left(\int_{\mathcal{\hat{W}}} \essinf_{\Pm_W} \Pm_{\hat{W}|W}\right) \\
          &= \left(\essinf_{\Pm_{\hat{W}}} \Pm_W\left(\ell(W,\hat{w})\right)\right)\exp\left(-\mathcal{L}^c(W\!\!\to\!\!{\hat{W}})\right),
      \end{align}
      where $\mathcal{L}^c(W\!\!\to\!\!{\hat{W}})$ represents maximal cost-leakage~\cite[Definition 11]{leakageLong}.
      \end{remark}
      \Cref{thm:sibsMIDirectResultBayesRisk} is different from the results presented in the previous section. While in~\Cref{sec:riskLBThroughProbs} the only dependence on $\ell$ was through the small-ball probability, in~\Cref{thm:sibsMIDirectResultBayesRisk} one is required to have access to the expected value of the $\beta$-th moments of $\ell$ with respect to $\Pm_X$. Such an object may not be as easy to bound as the small-ball probability.
	\begin{remark}
		If $W=\hat{W}$ then $\ell(W,\hat{W})=0$ and $I_\alpha(W,\hat{W})=0$. If   $\,0<\alpha<1$, given that $\beta<0$, one recovers the following lower-bound on the risk, which matches with our intuition: $
		\Pm_{W\hat{W}}(\ell(W,\hat{W})) \geq 0.$ 
	\end{remark}

	 The main difference with the results presented in~\Cref{sec:riskLBThroughProbs} is that there is no small-ball probability involved and it is thus required to have access to an object of the form $\min_{\hat{w}}\left\lVert \ell(W,\hat{w})\right\rVert_{L^\beta(\Pm_W)}$ (where we are abusing the notation as for $\beta<0$ it is not a norm), which might be harder to compute than $L_W(\rho)$. \iffalse It is, however, possible to relate the quantity to a small-ball probability, restricting the value of $\alpha$ to $[-\infty,0)$ and using a combination of Markov's inequality and~\Cref{thm:sibsIalphaBoundLess1Prob}. This approach leads us to the following lower-bound on the risk
	\begin{corollary}\label{thm:directLowerBoundBayesRiskMarkov}
		Consider the Bayesian framework described in~\Cref{sec:bayesianFramework}. The following must hold for every $\alpha<0$ and $\rho>0$:
		\begin{equation}
			R_B\geq \rho\left( \min_{\hat{w}}\Pm_W(\ell(W,\hat{w})\geq\rho)^{\frac1\beta}\cdot \exp\left(-\frac{\alpha-1}{\alpha}I_\alpha(W,X)\right)\right)
		\end{equation}
	\end{corollary}
	And, in particular,  \begin{align} \min_{\hat{w}}\Pm_W(\ell(W,\hat{w})\geq \rho)  &= \min_{\hat{w}}(1-\Pm_W(\ell(W,\hat{w})\leq \rho))\\ &=1-\max_{\hat{w}}\Pm_W(\ell(W,\hat{w})\leq \rho)\\
		&=1-L_W(\rho).\end{align}
	Hence, if we can upper-bound $L_W(\rho)$ like we assumed in the previous section, then we can lower-bound \\b$\min_{\hat{w}}\Pm_W(\ell(W,\hat{w})\geq \rho)$ and, consequently, we can also lower-bound the risk using the approach depicted in this section (cf.~\Cref{thm:directLowerBoundBayesRiskMarkov}).\fi

%\end{itemize}
 
 \bibliographystyle{IEEEtran}
	\bibliography{sample}

% Generated by IEEEtran.bst, version: 1.14 (2015/08/26)
\begin{thebibliography}{10}
\providecommand{\url}[1]{#1}
\csname url@samestyle\endcsname
\providecommand{\newblock}{\relax}
\providecommand{\bibinfo}[2]{#2}
\providecommand{\BIBentrySTDinterwordspacing}{\spaceskip=0pt\relax}
\providecommand{\BIBentryALTinterwordstretchfactor}{4}
\providecommand{\BIBentryALTinterwordspacing}{\spaceskip=\fontdimen2\font plus
\BIBentryALTinterwordstretchfactor\fontdimen3\font minus
  \fontdimen4\font\relax}
\providecommand{\BIBforeignlanguage}[2]{{%
\expandafter\ifx\csname l@#1\endcsname\relax
\typeout{** WARNING: IEEEtran.bst: No hyphenation pattern has been}%
\typeout{** loaded for the language `#1'. Using the pattern for}%
\typeout{** the default language instead.}%
\else
\language=\csname l@#1\endcsname
\fi
#2}}
\providecommand{\BIBdecl}{\relax}
\BIBdecl

\bibitem{bayesRiskRaginsky}
A.~{Xu} and M.~{Raginsky}, ``Information-theoretic lower bounds on bayes risk
  in decentralized estimation,'' \emph{IEEE Transactions on Information
  Theory}, vol.~63, no.~3, pp. 1580--1600, 2017.

\bibitem{nipsHideAndSeek}
O.~Shamir, ``Fundamental limits of online and distributed algorithms for
  statistical learning and estimation,'' in \emph{Advances in Neural
  Information Processing Systems}, Z.~Ghahramani, M.~Welling, C.~Cortes,
  N.~Lawrence, and K.~Q. Weinberger, Eds., vol.~27.\hskip 1em plus 0.5em minus
  0.4em\relax Curran Associates, Inc., 2014, pp. 163--171.

\bibitem{estimationVanTrees}
\BIBentryALTinterwordspacing
H.~L.~V. Trees, \emph{Detection, Estimation, and Modulation Theory}.\hskip 1em
  plus 0.5em minus 0.4em\relax John Wiley \& Sons, Ltd, 2001. [Online].
  Available: \url{https://onlinelibrary.wiley.com/doi/abs/10.1002/0471221082}
\BIBentrySTDinterwordspacing

\bibitem{estimationVanTrees1}
\BIBentryALTinterwordspacing
M.~Sato and M.~Akahira, ``An information inequality for the bayes risk,''
  \emph{The Annals of Statistics}, vol.~24, no.~5, pp. 2288--2295, 1996.
  [Online]. Available: \url{http://www.jstor.org/stable/2242654}
\BIBentrySTDinterwordspacing

\bibitem{estimationVanTrees2}
\BIBentryALTinterwordspacing
L.~D. Brown and L.~Gajek, ``Information inequalities for the bayes risk,''
  \emph{The Annals of Statistics}, vol.~18, no.~4, pp. 1578--1594, 1990.
  [Online]. Available: \url{http://www.jstor.org/stable/2241876}
\BIBentrySTDinterwordspacing

\bibitem{estimationVanTrees3}
H.~L. Van~Trees and K.~L. Bell, \emph{Bounds on the {B}ayes and {M}inimax
  {R}isk for Signal Parameter Estimation}, 2007, pp. 329--337.

\bibitem{estimationVanTrees4}
L.~Brown and R.~Liu, ``Bounds on the bayes and minimax risk for signal
  parameter estimation,'' \emph{IEEE Transactions on Information Theory},
  vol.~39, no.~4, pp. 1386--1394, 1993.

\bibitem{han}
{Te Sun Han} and S.~{Amari}, ``Statistical inference under multiterminal data
  compression,'' \emph{IEEE Transactions on Information Theory}, vol.~44,
  no.~6, pp. 2300--2324, 1998.

\bibitem{duchiEstimationIT}
Y.~Zhang, J.~Duchi, M.~I. Jordan, and M.~J. Wainwright, ``Information-theoretic
  lower bounds for distributed statistical estimation with communication
  constraints,'' in \emph{Advances in Neural Information Processing Systems},
  C.~J.~C. Burges, L.~Bottou, M.~Welling, Z.~Ghahramani, and K.~Q. Weinberger,
  Eds., vol.~26.\hskip 1em plus 0.5em minus 0.4em\relax Curran Associates,
  Inc., 2013.

\bibitem{duchiFano2013}
\BIBentryALTinterwordspacing
J.~C. Duchi and M.~J. Wainwright, ``Distance-based and continuum fano
  inequalities with applications to statistical estimation,'' \emph{CoRR}, vol.
  abs/1311.2669, 2013. [Online]. Available:
  \url{http://arxiv.org/abs/1311.2669}
\BIBentrySTDinterwordspacing

\bibitem{CalmonSDPIHockeyStick}
S.~Asoodeh, M.~Aliakbarpour, and F.~P. Calmon, ``Local differential privacy is
  equivalent to contraction of an $f$-divergence,'' in \emph{2021 IEEE
  International Symposium on Information Theory (ISIT)}, 2021, pp. 545--550.

\bibitem{bayesRiskFInformativity}
X.~Chen, A.~Guntuboyina, and Y.~Zhang, ``On bayes risk lower bounds,'' \emph{J.
  Mach. Learn. Res.}, vol.~17, no.~1, p. 7687–7744, jan 2016.

\bibitem{fInformativity}
I.~Csisz{\'a}r, ``A class of measures of informativity of observation
  channels,'' \emph{Periodica Mathematica Hungarica}, vol.~2, pp. 191--213,
  1972.

\bibitem{fullVersionGeneralization}
A.~R. Esposito, M.~Gastpar, and I.~Issa, ``Generalization error bounds via
  rényi-, f-divergences and maximal leakage,'' \emph{IEEE Transactions on
  Information Theory}, vol.~67, no.~8, pp. 4986--5004, 2021.

\bibitem{largeDeviationConvexAnalysis}
F.~Rassoul-Agha and T.~Seppäläinen, \emph{A course on large deviations with
  an introduction to Gibbs measures}, 05 2015.

\bibitem{RenyiKLDiv}
T.~van Erven and P.~Harremo\"es, ``{R}\'enyi divergence and
  {K}ullback-{L}eibler divergence,'' \emph{IEEE Trans. Inf. Theory}, vol.~60,
  no.~7, pp. 3797--3820, July 2014.

\bibitem{opMeanRDiv1}
I.~{Csiszar}, ``{G}eneralized cutoff rates and {R}\'enyi's information
  measures,'' \emph{IEEE Transactions on Information Theory}, vol.~41, no.~1,
  pp. 26--34, Jan 1995.

\bibitem{infoRadius}
R.~Sibson, ``Information radius,'' \emph{Z. Wahrscheinlichkeitstheorie verw
  Gebiete 14}, pp. 149--160, 1969.

\bibitem{verduAlpha}
S.~Verd{\'{u}}, ``{\(\alpha\)}-mutual information,'' in \emph{2015 Information
  Theory and Applications Workshop, {ITA} 2015, San Diego, CA, USA, February
  1-6, 2015}, 2015, pp. 1--6.

\bibitem{leakageLong}
I.~{Issa}, A.~B. {Wagner}, and S.~{Kamath}, ``An operational approach to
  information leakage,'' \emph{IEEE Transactions on Information Theory},
  vol.~66, no.~3, pp. 1625--1657, 2020.

\bibitem{fDiv1}
\BIBentryALTinterwordspacing
F.~Liese and I.~Vajda, ``On divergences and informations in statistics and
  information theory,'' \emph{IEEE Trans. Inf. Theor.}, vol.~52, no.~10, pp.
  4394--4412, 2006. [Online]. Available:
  \url{http://dx.doi.org/10.1109/TIT.2006.881731}
\BIBentrySTDinterwordspacing

\bibitem{sdpiRaginsky}
M.~Raginsky, ``Strong data processing inequalities and $\phi $-{S}obolev
  inequalities for discrete channels,'' \emph{IEEE Transactions on Information
  Theory}, vol.~62, no.~6, pp. 3355--3389, 2016.

\bibitem{cohenDobrushin}
\BIBentryALTinterwordspacing
J.~E. Cohen, Y.~Iwasa, G.~Rautu, M.~{Beth Ruskai}, E.~Seneta, and G.~Zbaganu,
  ``Relative entropy under mappings by stochastic matrices,'' \emph{Linear
  Algebra and its Applications}, vol. 179, pp. 211--235, 1993. [Online].
  Available:
  \url{https://www.sciencedirect.com/science/article/pii/002437959390331H}
\BIBentrySTDinterwordspacing

\bibitem{SDPIfDiv}
A.~Makur and L.~Zheng, ``Comparison of contraction coefficients for
  f-divergences,'' \emph{Problems of Information Transmission}, vol.~56, p.
  103–156, 2020.

\bibitem{largeDeviationVaradhan}
S.~Varadhan, \emph{Large Deviations and Applications}, 1984.

\bibitem{anantharamRenyiDivVarRepr}
\BIBentryALTinterwordspacing
V.~Anantharam, ``A variational characterization of r{\'{e}}nyi divergences,''
  \emph{CoRR}, vol. abs/1701.07796, 2017. [Online]. Available:
  \url{http://arxiv.org/abs/1701.07796}
\BIBentrySTDinterwordspacing

\bibitem{RenyiDivVarReprNeuralNetwork}
J.~Birrell, P.~Dupuis, M.~A. Katsoulakis, L.~Rey-Bellet, and J.~Wang,
  ``Variational representations and neural network estimation of r\'enyi
  divergences,'' 2021.

\bibitem{minimizationMeasures}
M.~Broniatowski and A.~Keziou, ``Minimization of divergences on sets of signed
  measures,'' \emph{Studia Scientiarum Mathematicarum Hungarica}, vol.~43,
  no.~4, pp. 403--442, 2006.

\bibitem{thesis}
\BIBentryALTinterwordspacing
A.~R. Esposito, ``A functional perspective on information measures,'' p. 170,
  2022. [Online]. Available: \url{http://infoscience.epfl.ch/record/294547}
\BIBentrySTDinterwordspacing

\bibitem{fDivegerceInequalities}
I.~Sason and S.~Verdú, ``$f$ -divergence inequalities,'' \emph{IEEE
  Transactions on Information Theory}, vol.~62, no.~11, pp. 5973--6006, 2016.

\bibitem{polyianksiyConverse}
Y.~Polyanskiy, H.~V. Poor, and S.~Verdu, ``Channel coding rate in the finite
  blocklength regime,'' \emph{IEEE Transactions on Information Theory},
  vol.~56, no.~5, pp. 2307--2359, 2010.

\bibitem{conditionalSibsMI}
A.~R. Esposito, D.~Wu, and M.~Gastpar, ``On conditional sibson's
  $\alpha$-mutual information,'' in \emph{2021 IEEE International Symposium on
  Information Theory (ISIT)}, 2021, pp. 1796--1801.

\bibitem{etaTVKL}
\BIBentryALTinterwordspacing
J.~E. Cohen, Y.~Iwasa, G.~Rautu, M.~{Beth Ruskai}, E.~Seneta, and G.~Zbaganu,
  ``Relative entropy under mappings by stochastic matrices,'' \emph{Linear
  Algebra and its Applications}, vol. 179, pp. 211--235, 1993. [Online].
  Available:
  \url{https://www.sciencedirect.com/science/article/pii/002437959390331H}
\BIBentrySTDinterwordspacing

\bibitem{rockafellar-1970a}
R.~T. Rockafellar, \emph{Convex analysis}, ser. Princeton Mathematical
  Series.\hskip 1em plus 0.5em minus 0.4em\relax Princeton, N. J.: Princeton
  University Press, 1970.

\bibitem{concrete_mathematics}
R.~L. Graham, D.~E. Knuth, and O.~Patashnik, \emph{Concrete Mathematics: A
  Foundation for Computer Science}.\hskip 1em plus 0.5em minus 0.4em\relax
  Reading: Addison-Wesley, 1989.

\bibitem{levyLittleSums}
S.~G. Bobkov and G.~P. Chistyakov, ``On concentration functions of random
  variables,'' \emph{Journal of Theoretical Probability volume}, vol.~28, 2015.

\bibitem{smallBallai}
H.~H. Nguyen and V.~H. Vu, \emph{Small Ball Probability, Inverse Theorems, and
  Applications}.\hskip 1em plus 0.5em minus 0.4em\relax Berlin, Heidelberg:
  Springer Berlin Heidelberg, 2013, pp. 409--463.

\bibitem{sibsonIalphaNegative}
\BIBentryALTinterwordspacing
A.~R. Esposito, A.~Vandenbroucque, and M.~Gastpar, ``On sibson's
  $\upalpha$-mutual information,'' in \emph{2022 {IEEE} International Symposium
  on Information Theory ({ISIT})}.\hskip 1em plus 0.5em minus 0.4em\relax
  {IEEE}, Jun. 2022. [Online]. Available:
  \url{https://doi.org/10.1109/isit50566.2022.9834428}
\BIBentrySTDinterwordspacing

\end{thebibliography}

	% biography section
	% 
	% If you have an EPS/PDF photo (graphicx package needed) extra braces are
	% needed around the contents of the optional argument to biography to prevent
	% the LaTeX parser from getting confused when it sees the complicated
	% \includegraphics command within an optional argument. (You could create
	% your own custom macro containing the \includegraphics command to make things
	% simpler here.)
	%\begin{IEEEbiography}[{\includegraphics[width=1in,height=1.25in,clip,keepaspectratio]{mshell}}]{Michael Shell}
	% or if you just want to reserve a space for a photo:
	
	% You can push biographies down or up by placing
	% a \vfill before or after them. The appropriate
	% use of \vfill depends on what kind of text is
	% on the last page and whether or not the columns
	% are being equalized.
	
	%\vfill
	
	% Can be used to pull up biographies so that the bottom of the last one
	% is flush with the other column.
	%\enlargethispage{-5in}

	% that's all folks
\end{document}